\documentclass{amsart}
\newtheorem{theorem}{Theorem}[section]

\theoremstyle{definition}

\newtheorem{definition}{Definition}
\newtheorem{corollary}{Corollary}
\newtheorem{proposition}{Proposition}
\newtheorem{example}{Example}
\theoremstyle{remark}
\newtheorem{remark}[theorem]{Remark}
\usepackage{ragged2e}
\usepackage{amsmath}
\numberwithin{equation}{section}
\usepackage[table, svgnames, dvipsnames]{xcolor}
\usepackage{makecell, cellspace, caption}
\usepackage{tabularx, ragged2e, booktabs, caption}
\usepackage{amssymb}
\usepackage{amsmath}
\usepackage{color,soul}
\usepackage{enumitem}
\usepackage{tabularx,ragged2e,booktabs,caption}
\usepackage{longtable}
\usepackage{graphicx}
\usepackage[colorlinks]{hyperref}

\begin{document}
	\title{Hulls of Free Linear Codes over a Non-Unital Ring }
	\author{Anup Kushwaha}
    \address{Department of Mathematics, Indian Institute of Technology Patna, Patna-801106}
    \curraddr{}
    \email{E-mail: anup$\textunderscore$2221ma11@iitp.ac.in}
    \thanks{}

   \author{ Om Prakash$^*$}
   \address{Department of Mathematics, Indian Institute of Technology Patna, Patna-801106}
   \curraddr{}
   \email{om@iitp.ac.in}
   \thanks{* Corresponding author}
	
	
	
	\dedicatory{}

\begin{abstract}
This paper investigates the hull codes of free linear codes over a non-unital ring $ E= \langle \kappa,\tau \mid 2 \kappa =2 \tau=0,~ \kappa^2=\kappa,~ \tau^2=\tau,~ \kappa \tau=\kappa,~ \tau \kappa=\tau \rangle$. Initially, we examine the residue and torsion codes of various hulls of $E$-linear codes and obtain an explicit form of the generator matrix of the hull of a free $E$-linear code. Then, we propose four build-up construction methods to construct codes with a larger length and hull-rank from codes with a smaller length and hull-rank. Some illustrative examples are also given to support our build-up construction methods. Subsequently, we study the permutation equivalence of two free $E$-linear codes and discuss the hull-variation problem. As an application, we classify optimal free $E$-linear codes for lengths up to $8$.
\end{abstract}

\maketitle
\section{Introduction}
Let $\mathbb{F}_q$ be the finite field with $q$ elements. A subspace $C$ of $\mathbb{F}_q^n$ refers to an $\mathbb{F}_q$-linear code of length $n$. Under the Euclidean inner product, the dual of the code $C$ is the collection of all vectors of $\mathbb{F}_q^n$ that are orthogonal to $C$ and denoted by $C^\perp$. The hull of the code $C$ is defined as $Hull(C)=C \cap C^\perp$. In fact, the concept of the hull of an $\mathbb{F}_q$-linear code was introduced in 1990 by Assmus and Key \cite{Assmus} to investigate certain properties of finite projective planes. Recall that two codes $C$ and $D$ are permutation-equivalent if $C$ can be derived from $D$ by applying a suitable coordinate permutation. Determining whether two linear codes are permutation-equivalent is a significant problem in coding theory. Recently, hulls have attracted intense research interest due to their various applications. In particular, the dimension of the hull directly influences the complexity of algorithms that check the permutation equivalence or determine the automorphism group of a linear code \cite{Leon,Send1997,Send2000}. Additionally, hulls can be utilized to construct quantum error-correcting codes \cite{Dougherty,Liu2}. These facts highlight the necessity of studying hull codes. Over the years, researchers have investigated the hulls of various codes over finite fields or finite rings with unity. However, there is no work available in the literature on the hull of codes over non-unital rings. Thus, the study of hulls over non-unital rings is a significant problem, and this is a first attempt towards this direction. \par

Recall that an $\mathbb{F}_q$-linear code $C$ is referred to as a linear complementary dual (LCD), self-orthogonal (SO) and self-dual (SD) code if it satisfies $C \cap C^\perp = \{0\}$, $C \subseteq C^\perp$ and $C=C^\perp$, respectively. We see that, if $Hull(C)$ is trivial, it is LCD, if $Hull(C)=C$, it is SO, and if $Hull(C)=C=C^\perp$, it is SD. Thus, the hull code can be seen as a basic generalization of LCD, SO  and SD codes. The concept of LCD codes was introduced by Massey in \cite{Massey92}, where he also proved the existence of asymptotically good LCD codes. In recent years, these codes have attracted considerable interest because of their use in countering Side-Channel Attacks \cite{Carl16}. Additionally, LCD codes have applications in secret sharing schemes \cite{Yadav2025}. Using the concept of hull dimension spectra, Sendrier proved in 2004 that LCD codes attain the Gilbert-Varshamov bound \cite{Send04}. Subsequently, Islam et al. \cite{Islam2022} demonstrated a method to construct both quantum and LCD codes leveraging cyclic codes over non-chain rings. Further, double circulant LCD and SD codes were investigated in \cite{Prakash}. On the other hand, research on SO and SD codes has seen significant growth due to their use in different areas such as quantum stabilizer codes \cite{Calderbank,Steane}, modular forms \cite{Minjia}, lattice theory \cite{Bannai,Conway,Harada}, and combinatorial $t$-design theory \cite{Bachoc}.  \par

Recently, there has been a significant shift towards considering non-unital rings to study different types of codes. In 2021, Alahmadi et al. \cite{Alah21} studied Quasi Type IV codes over a non-unital ring. Then, in 2022, Shi et al. \cite{Shi21} were the first to consider a non-unitary ring with $4$ elements to investigate LCD codes. Following these, Kim and Roe \cite{Kim2022a} investigated quasi-self-dual (QSD) codes over a non-unitary ring with 4 elements. Note that a QSD code $C$ of length $n$ over a finite ring $R$ is an SO code with $|C|=|R|^{n/2}$. Subsequently, authors in \cite{Kushwaha1} extended QSD and LCD codes over a non-unitary ring with $9$ elements. Meanwhile, Deb et al. \cite{Deb}  worked on classifying certain codes over a non-unitary ring of order $4$. This was followed by the work of Kushwaha et al. \cite{Kushwaha2}, who classified MDS and almost MDS LCD and SD codes over a non-unital ring $E_p$. On the other hand, the build-up construction method is a powerful technique to construct codes with a larger length from codes with a smaller length. This technique was utilized in \cite{Alah22} and \cite{Kim2022a} for the classification of Type IV and QSD codes over non-unital rings for some smaller lengths. Further, binary linear codes with various hull dimensions were classified in \cite{Li} using build-up construction techniques. In 2020, Liu and Pan \cite{Liu} studied the hull-variation problem for $\mathbb{F}_q$-linear codes, while LCD and $\mathbb{F}_q$-linear codes with hull-dimension $1$ were investigated in \cite{Wang}. Motivated by these works, we study the hulls of free codes over the non-unital ring $ E= \langle \kappa,\tau \mid 2 \kappa =2 \tau=0,~ \kappa^2=\kappa,~ \tau^2=\tau,~ \kappa \tau=\kappa,~ \tau \kappa=\tau \rangle$. To achieve this, we explore the various hulls of the $E$-linear codes and identify their residue and torsion codes. Then, we produce the form of the generator matrix of the hull code of a free $E$-linear code and also determine its hull-rank. Subsequently, we present four build-up construction methods that utilize free $E$-linear codes with a smaller length and hull-rank to construct free $E$-linear codes with a larger length and hull-rank. We also give some supporting examples of codes constructed by using these build-up construction methods. In addition, we investigate the permutation equivalence of $E$-linear codes and discuss the hull-variation problem for free $E$-linear codes. We conclude our work by classifying optimal free $E$-linear codes for lengths up to $8$. \par

We organize this work in the following manner. Section 2 provides some preliminaries and studies various hulls of an $E$-linear code. It also deals with the generator matrix of the hull code of a free $E$-linear code and its hull-rank. In Section $3$, the four build-up construction methods are presented. Section $4$ investigates the hull-variation problem for free $E$-linear and classifies the optimal free $E$-linear codes for lengths up to $8$. We conclude our work in Section $5$.

\section{Hulls over the ring $E$}
In \cite{Fine93}, Fine has classified rings of order $4$ up to isomorphism. There are only two noncommutative non-unitay rings with $4$ elements, namely, the ring $ E= \langle \kappa,\tau \mid 2 \kappa =2 \tau=0,~ \kappa^2=\kappa,~ \tau^2=\tau,~ \kappa \tau=\kappa,~ \tau \kappa=\tau \rangle$ and its opposite ring. This paper considers the ring $E$ as a code alphabet to study the hull of free linear codes over $E$. One can derive similar results over the opposite ring. The characteristic of the ring $E$ is $2$ and it has $4$ elements $\{ ~i \kappa+j \tau~ |~ 0 \leq i,j <2 ~\}$. We denote $\kappa + \tau=\zeta$ and give the multiplication table of the ring $E$ in Table \ref{Tab1} to characterize its structure.

\begin{table}[ht]
\centering
\caption{\label{Tab1} Multiplication table of $E$.}
\begin{tabular}{|c|c|c|c|c|}
\hline
  $\cdot$ & $0$ & $\kappa$ & $\tau$ &  $\zeta$ \\
 \hline
 $0$ & $0$ & $0$ & $0$ & $0$  \\

\hline
 $\kappa$ & $0$ & $\kappa$ & $\kappa$ &  $0$ \\
 \hline
 $\tau$ & $0$ & $\tau$ & $\tau$ & $0$
 \\
 
 \hline
 $\zeta$ & $0$ & $\zeta$ & $\zeta$ & $0$
 \\
 
\hline
\end{tabular}
\end{table}

The multiplication table shows that the ring $E$ is non-unitary and has a unique maximal ideal $J=\{0, \zeta \}$.  Hence, $E$ is a local ring. Therefore, its residue field is given by $E/J=\mathbb{F}_2$. Further, $x \kappa=x \tau=x$ for all $ x \in E$. Also, every element $x \in E$ has a $\zeta$-adic decomposition as follows:
 $$x=u \kappa+v \zeta ~~ \text{where}~ u,v \in \mathbb{F}_2.$$
 Now, for all $x \in E$ and $v \in \mathbb{F}_2$, an action of $\mathbb{F}_2$ on $E$ can be defined as $xv=vx$. We see that this action is distributive, i.e., $x(u\oplus_{\mathbb{F}_2} v)=xu+xv=ux+vx$ for all $ x \in E$ and $u,v \in \mathbb{F}_2$. Now, if an element of $E$ is written in the form of $\zeta$-adic decomposition, then the map of reduction modulo $J$ is a  map $\pi:E\rightarrow E/J =\mathbb{F}_2$ defined by
 $$\pi(x)=\pi(u\kappa+v\zeta)=u.$$
 There is a natural extension of this map from $E^n$ to  $\mathbb{F}_2^n$.

 \begin{remark}
        In this paper, we fix the notations $\kappa$ and $\tau$ for the generators of the ring $E$ that satisfy the conditions $ 2 \kappa =2 \tau=0,~ \kappa^2=\kappa,~ \tau^2=\tau,~ \kappa \tau=\kappa,~ \tau \kappa=\tau $. Moreover,  $\zeta=\kappa+ \tau.$
        
     \end{remark}

\begin{definition}[Linear code]\label{def1}
     An $E$-linear code of length $n$ refers to a left $E$-submodule of $E^n$.
\end{definition}

\begin{definition}[Minimum distance]
The number of differing components between two codewords $w$ and $z$ is called the distance between them and is denoted by $d(w,z)$. Then, the minimum distance of the code $C$ is
$$d(C)=\text{min}\{d(w,z)~|~ w,z \in C, w \neq z\}.$$
 
\end{definition}

  Here, given an $E$-linear code $C$ of length $n$, the following two $\mathbb{F}_2$-linear codes related to $C$ are defined.
\begin{enumerate}
\item \textbf{Residue code:} For an $E$-linear code $C$, $\pi(C)$ is called its residue code and denoted by $C_{Res}$, i.e., 
    $$ C_{Res}=\{\pi(z)~ |~ z \in C \}.$$
    \item \textbf{Torsion code:} For an $E$-linear code $C$, its torsion code is an $\mathbb{F}_2$-linear code given by
    $$ C_{Tor}=\{ v \in \mathbb{F}_{2}^{n} ~ |~ v \zeta \in C\}.$$
    
    \end{enumerate}

    Now, we define the inner product of any vectors  $w=(w_1, w_2, \ldots, w_n)$, and $z=(z_1,z_2,\ldots,z_n)$ in $E^n$ as
    $$ \langle w, z  \rangle=\sum_{j=1}^{n}w_j z_j.$$
    
   Additionally, the duals of an $E$-linear code $C$ are introduced below, defined under the inner product given above.
    \begin{enumerate}
     \item \textbf{Left dual:} The $E$-linear code given by 
    $$C^{\perp_L}= \{ \boldsymbol{z} \in E^n~ |~\langle \boldsymbol{z}, \boldsymbol{w}  \rangle= 0,\forall ~\boldsymbol{w}\in C\},$$
is said to be the left dual of the code $C$.
    \item \textbf{Right dual:} The $E$-linear code given by 
        $$C^{\perp_R}= \{ \boldsymbol{z }\in E^n~ |~\langle \boldsymbol{w}, \boldsymbol{z}  \rangle= 0,\forall ~\boldsymbol{w} \in C\},$$
        is defined as the right dual of the code $C$.
        
        \item \textbf{Dual:} The intersection of left and right duals of an $E$-linear code $C$ is defined as the dual of the code $C$ and denoted by $C^\perp$, i.e., $C^\perp=C^{\perp_L} \cap C^{\perp_R}$.
    \end{enumerate}

\begin{definition}[Left hull, right hull and hull codes]
  For an $E$-linear code $C$, its left hull, right hull and hull codes are defined as $LHull(C)=C\cap C^{\perp_L}$, $RHull(C)=C\cap C^{\perp_R}$ and $Hull(C)=C\cap C^{\perp}$, respectively.  
\end{definition}

\begin{definition}[Generating set, generator matrix and parity-check matrix] For an $E$-linear code $C$, let  $X=\{  x_1,x_2,\ldots , x_k \} \subset C$. Then the set $$\langle X \rangle_{E} = \{ a_1x_1+a_2x_2+ \cdots + a_kx_k ~|~ a_i \in E, ~  1 \leq i \leq k \},$$
 is called the (left) $E$-span of $X$. Next, 
 the set given by
$$\langle X \rangle_{\mathbb{F}_2} = \{ v_1x_1+v_2x_2+ \cdots + v_kx_k ~|~ v_i \in \mathbb{F}_2, ~ 1 \leq i \leq k  \},$$
is called the additive span of the set $X$.
Since the ring $E$ is non-unitary, the additive span of the set $X$ is not always contained in the $E$-span of $X$. Now, if the set $X=\{ x_1,x_2,\ldots , x_k \}$ satisfies
    $$\langle X \rangle_{E} \cup \langle X \rangle_{\mathbb{F}_2}=C,$$
    it is termed as a generating set for the code $C$. Next, if $X=\{ x_1,x_2,\ldots, x_k \}\subset C$ is a generating set of an $E$-linear code $C$ of length $n$, then a $k \times n$ matrix $G_{E}$ whose rows are $x_1,x_2,\ldots, x_k$ and $\langle G \rangle_{E}=\langle X \rangle_{E} \cup \langle X \rangle_{\mathbb{F}_2}$ refers to a generator matrix  of $C$. Further, a generator matrix of $C^\perp$ is termed the parity-check matrix of $C$.

\end{definition}

\begin{definition}[Free code] If an $E$-linear code $C$ can be decomposes as 
$$C=E\oplus E\oplus \cdots \oplus E,$$ where each component $E$ is generated by some $z_i \in E$ as an $E$-module, it is called free. Note that the freeness of the $E$-linear code $C$ is equivalent to the condition $C_{Res} = C_{Tor}$.
\end{definition}

\begin{definition}[Rank of a free code]
   For a free $E$-linear code $C$, its rank, denoted by $rank(C)$, is the size of its minimal generating set. 
\end{definition}

 We now state a few basic results derived from Lemma $1$, Theorem $9$ and Theorem $10$ of \cite{Alah23}.

\begin{theorem}\label{Thm1b}
If $C$ is an $E$-linear code of length $n$, then
    
\begin{enumerate}
    \item $\kappa C_{Res} \subseteq C$ and $\zeta C_{Tor} \subseteq C$,
    \item $(C^{\perp_L})_{Res}=(C_{Res})^\perp=(C^{\perp_L})_{Tor}$,
    \item $(C^{\perp_R})_{Res}=(C_{Tor})^\perp$ and $(C^{\perp_R})_{Tor}=\mathbb{F}_2^n$,
    \item $(C^{\perp})_{Res}=(C_{Tor})^\perp$ and $(C^{\perp})_{Tor}=(C_{Res})^\perp$.

    \end{enumerate}
\end{theorem}

 Now, we have the following result to check the freeness of the three duals of an $E$-linear code.

 \begin{theorem}
    If $C$ is an $E$-linear code, then $C^{\perp_L}$ is always free but $C^{\perp_R}$ is never free. Further, $C^{\perp}$ is free if and only if $C$ is free.
 \end{theorem}
\begin{proof}
  Freeness of $C^{\perp_L}$ and non-freeness of $C^{\perp_R}$ follow the Lemma $5$ of \cite{Shi21}. To prove the necessity of the second result, i.e., for the free $E$-linear code $C$, the result (4) of Theorem \ref{Thm1b} implies that $$(C^\perp)_{Res}=(C_{Res})^\perp=(C^\perp)_{Tor}.$$
  Therefore, $C^\perp$ is free. On the other hand, assume $C^\perp$ is a free $E$-linear code. Since $(C^\perp)_{Res}=(C^\perp)_{Tor}$, again (4) of Theorem \ref{Thm1b} implies that  $(C_{Res})^\perp=(C_{Tor})^\perp$. Thus, $C_{Res}=C_{Tor}$, and so $C$ is free.
\end{proof}

The following result calculates the residue and torsion codes of the three hull codes of an $E$-linear code.

\begin{theorem}\label{Thm2}
    Let $C$ be an $E$-linear code. Then
\begin{enumerate}
    \item $(LHull(C))_{Res}=C_{Res}\cap (C_{Res})^\perp$ and $(LHull(C))_{Tor}=(C_{Res})^\perp \cap C_{Tor}$,
    \item $(RHull(C))_{Res}=C_{Res} \cap (C_{Tor})^\perp$ and $(RHull(C))_{Tor}=C_{Tor}$,
    \item $(Hull(C))_{Res}=C_{Res} \cap (C_{Tor})^\perp$ and $(Hull(C))_{Tor}=(C_{Res})^\perp \cap C_{Tor}$.
\end{enumerate}
\end{theorem}
\begin{proof}  \begin{enumerate}
    \item Let $u \in (LHull(C))_{Res}$. Then, by (1) of Theorem \ref{Thm1b}, $\kappa u \in C \cap C^{\perp_L}$ which implies that $\kappa u \in C$ and $\kappa u \in C^{\perp_L}$. Since $\pi(\kappa u)=u$ and $(C^{\perp_L})_{Res}=(C_{Res})^\perp$ by (2) of Theorem \ref{Thm1b}, we have  $u \in C_{Res}$ and $u \in (C_{Res})^\perp$. Hence, $u \in C_{Res} \cap (C_{Res})^\perp$. Therefore, $(C \cap C^{\perp_L})_{Res} \subseteq C_{Res} \cap (C_{Res})^\perp$. For the converse inclusion, let $v \in C_{Res} \cap (C_{Res})^\perp$. This implies that $v \in C_{Res}$ and $v \in  (C_{Res})^\perp= (C^{\perp_L})_{Res}$. Then, by (1) of Theorem \ref{Thm1b}, $\kappa v \in C$ and $\kappa v \in C^{\perp_L}$. Hence, $\kappa v \in C \cap C^{\perp_L}$. Since $\pi(\kappa v)=v$, we have $v \in (C \cap C^{\perp_L})_{Res}$. Therefore, $C_{Res} \cap (C_{Res})^\perp \subseteq (C \cap C^{\perp_L})_{Res}$. Thus, $(LHull(C))_{Res}=C_{Res}\cap (C_{Res})^\perp$.\par
   
    On the other hand, let $ u \in (LHull(C))_{Tor}$. Then, by (1) of Theorem \ref{Thm1b}, $\zeta u \in C \cap C^{\perp_L}$. This implies that $\zeta u \in C $ and $\zeta u \in C^{\perp_L}$. Hence, $u \in C_{Tor}$ and $u \in (C^{\perp_L})_{Tor}$. Also, by (2) of Theorem \ref{Thm1b}, $u \in C_{Tor} \cap (C_{Res})^\perp$. Therefore, $(LHull(C))_{Tor} \subseteq C_{Tor} \cap (C_{Res})^\perp$. For the converse inclusion, let $v \in C_{Tor} \cap (C_{Res})^\perp$. Then, by (2) of Theorem \ref{Thm1b}, $v \in C_{Tor}$ and $v \in (C^{\perp_L})_{Tor}$. Also, by (1) of Theorem \ref{Thm1b}, $\zeta v \in C$ and $\zeta v \in C^{\perp_L}$. Consequently, $\zeta v \in C \cap C^{\perp_L}$. Hence, $v \in (C \cap C^{\perp_L})_{Tor}$. Therefore, $C_{Tor} \cap (C_{Res})^\perp \subseteq (C \cap C^{\perp_L})_{Tor}$. Thus, $(LHull(C))_{Tor}=(C_{Res})^\perp \cap C_{Tor}$.

      \item The residue code of the right hull can be obtained by the following procedure, analogous to that used for deriving the residue code of the left hull. To derive the torsion code of the right hull, let $ u \in (RHull(C))_{Tor}$. Then, by (1) of Theorem \ref{Thm1b}, $\zeta u \in C \cap C^{\perp_R}$. This implies that $\zeta u \in C $ and $\zeta u \in C^{\perp_R}$. Hence, $u \in C_{Tor}$ and $u \in (C^{\perp_R})_{Tor}$. Also, (3) of Theorem \ref{Thm1b} implies that $u \in C_{Tor} \cap \mathbb{F}_2^n=C_{Tor}$. Therefore, $(RHull(C))_{Tor} \subseteq C_{Tor}$. For the converse inclusion, let $v \in C_{Tor}$. Since  $v \in \mathbb{F}_2^n$, by (3) of Theorem \ref{Thm1b}, $v \in (C^{\perp_R})_{Tor}$. Also, by (1) of Theorem \ref{Thm1b}, $\zeta v \in C$ and $\zeta v \in C^{\perp_R}$. Consequently, $\zeta v \in C \cap C^{\perp_R}$. Hence, $v \in (C \cap C^{\perp_R})_{Tor}$. Therefore, $C_{Tor}  \subseteq (C \cap C^{\perp_R})_{Tor}$. Thus, $(RHull(C))_{Tor}= C_{Tor}$.
     
     \item This part follows a similar procedure to the above two parts.

\end{enumerate}
\end{proof} 

An important consequence of the above theorem is the following corollary.
\begin{corollary}\label{Cor1}
   If $C$ is an $E$-linear code, then
   $$(Hull(C))_{Res} \subseteq Hull(C_{Res}) ~\text{and}~~ (Hull(C))_{Tor} \supseteq Hull(C_{Tor}).$$ 
Furthermore, if $C$ is free, equality holds.
\end{corollary}
\begin{proof}
   For an $E$-linear code  $C$, by $(3)$ of Theorem \ref{Thm2}, we have 
   $$(Hull(C))_{Res}=C_{Res} \cap (C_{Tor})^\perp ~~\text{and}~~ (Hull(C))_{Tor}=(C_{Res})^\perp \cap C_{Tor}.$$
   Since $C$ is linear, $C_{Res} \subseteq C_{Tor}$. This implies that $(C_{Tor})^\perp \subseteq (C_{Res})^\perp$. Hence, $C_{Res} \cap (C_{Tor})^\perp \subseteq C_{Res} \cap (C_{Res})^\perp$. Therefore, $(Hull(C))_{Res} \subseteq Hull(C_{Res})$. Next, $C_{Tor} \supseteq C_{Res}$ implies that $(C_{Res})^\perp \cap C_{Tor} \supseteq (C_{Res})^\perp \cap C_{Res}$. Therefore, $(Hull(C))_{Tor} \supseteq Hull(C_{Tor})$. On the other hand, if $C$ is free, then $C_{Res}=C_{Tor}$ implies that  
$$(Hull(C))_{Res}=C_{Res} \cap (C_{Tor})^\perp=C_{Res} \cap (C_{Res})^\perp=Hull(C_{Res}),$$
and 
     $$(Hull(C))_{Tor}=(C_{Res})^\perp \cap C_{Tor}=(C_{Tor})^\perp \cap C_{Tor}=Hull(C_{Tor}).$$
   
\end{proof}

We know that if $C$ is an $\mathbb{F}_q$-linear code, then $(C^\perp)^\perp=C$. Consequently, we have the following results for the three duals of an $E$-linear code.

\begin{theorem}\label{Thm4e}
   An $E$-linear code $C$ of length $n$ satisfies the following duality properties:
    \begin{enumerate}
        \item $(C^{\perp_L})^{\perp_L}=C$ if $C$ is free; 
        \item $(C^{\perp_R})^{\perp_R}=C$ if  $C_{Res}=\{0\}$ and $C_{Tor}=\mathbb{F}_2^n$; 
        \item $(C^\perp)^\perp=C$ always;
        \item $(C^{\perp_L})^{\perp_R}=C$ if $C_{Tor}=\mathbb{F}_2^n$;
        \item $(C^{\perp_R})^{\perp_L}=C$ if $C$ is free.
        \end{enumerate}
\end{theorem}
\begin{proof}
    \begin{enumerate}
        \item If $C$ is an $E$-linear code, then by (2) of Theorem \ref{Thm1b}, we have 
        $$((C^{\perp_L})^{\perp_L})_{Res}=((C^{\perp_L})_{Res})^\perp=((C_{Res})^\perp)^\perp=C_{Res},$$
        and 
        $$((C^{\perp_L})^{\perp_L})_{Tor}=((C^{\perp_L})_{Res})^\perp=((C_{Res})^\perp)^\perp=C_{Res}.$$
       If $C$ is free, then $C_{Res}=C_{Tor}$ implies that $((C^{\perp_L})^{\perp_L})_{Res}= C_{Res}$ and $((C^{\perp_L})^{\perp_L})_{Tor}=C_{Tor}$. Therefore, $(C^{\perp_L})^{\perp_L}=C$. This completes the proof.

        \item If $C$ is an $E$-linear code, then by (3) of Theorem \ref{Thm1b}, we have 
        $$((C^{\perp_R})^{\perp_R})_{Res}=((C^{\perp_R})_{Tor})^\perp=(\mathbb{F}_2^n)^\perp=\{0\} ~~\text{and}~~ ((C^{\perp_R})^{\perp_R})_{Tor}=\mathbb{F}_2^n.$$
        We know that $(C^{\perp_R})^{\perp_R}=C$ if $((C^{\perp_R})^{\perp_R})_{Res}= C_{Res}$ and $((C^{\perp_R})^{\perp_R})_{Tor}=C_{Tor}$. Therefore, $(C^{\perp_R})^{\perp_R}=C$ if $C_{Res}=\{0\}$ and $C_{Tor}=\mathbb{F}_2^n$.

        \item For an $E$-linear code $C$, (4) of Theorem \ref{Thm1b} implies that
        $$((C^{\perp})^{\perp})_{Res}=((C^{\perp})_{Tor})^\perp=((C_{Res})^\perp)^\perp=C_{Res},$$
        and 
        $$((C^{\perp})^{\perp})_{Tor}=((C^{\perp})_{Res})^\perp=((C_{Tor})^\perp)^\perp=C_{Tor}.$$
        Since $((C^{\perp_L})^{\perp_L})_{Res}= C_{Res}$ and $((C^{\perp_L})^{\perp_L})_{Tor}=C_{Tor}$, we have $(C^\perp)^\perp=C$. 
        
        \item For an $E$-linear code $C$, (2) and (3) of Theorem \ref{Thm1b} imply that 
        $$((C^{\perp_L})^{\perp_R})_{Res}=((C^{\perp_L})_{Tor})^\perp=((C_{Res})^\perp)^\perp=C_{Res}~~\text{and}~~((C^{\perp_L})^{\perp_R})_{Tor}=\mathbb{F}_2^n.$$
         Therefore, $(C^{\perp_L})^{\perp_R}=C$ if  $C_{Tor}=\mathbb{F}_2^n$.

        \item If $C$ is an $E$-linear code, then (2) and (3) of Theorem \ref{Thm1b} imply that 
        $$((C^{\perp_R})^{\perp_L})_{Res}=((C^{\perp_R})_{Res})^\perp=((C_{Tor})^\perp)^\perp=C_{Tor},$$
        and 
        $$((C^{\perp_R})^{\perp_L})_{Tor}=((C^{\perp_R})_{Res})^\perp=((C_{Tor})^\perp)^\perp=C_{Tor}.$$
         Therefore, $(C^{\perp_R})^{\perp_L}=C$ if  $C_{Res}=C_{Tor}$. Thus, $(C^{\perp_R})^{\perp_L}=C$ if $C$ is free. 
    \end{enumerate}
\end{proof}

We know that if $C$ is an $\mathbb{F}_q$-linear code, then $Hull(C)=Hull(C^\perp)$. Correspondingly, we have the following result for three hulls of an $E$-linear code.

\begin{theorem}
    An $E$-linear code $C$ of length $n$ satisfies the following properties:
    \begin{enumerate}
        \item $LHull(C)=LHull(C^{\perp_L})$ if $C$ is free, 
        \item  $Hull(C)=Hull(C^\perp)$ always,
        \item $RHull(C)=RHull(C^{\perp_R})$ if $C_{Res} \cap (C_{Tor})^\perp=\{0\}$ and $C_{Tor}=\mathbb{F}_2^n$.
  \item $RHull(C)=LHull(C^{\perp_R})$ if $C$ is free,
  \item $LHull(C)=RHull(C^{\perp_L})$ if $C_{Res}$ is SD and $C$ is free.
\end{enumerate}
\end{theorem}
\begin{proof}
    \begin{enumerate}
        \item If $C$ is a free $E$-linear code, then (1) of Theorem \ref{Thm2} implies that
        \begin{align*}
        (LHull(C^{\perp_L}))_{Res}&=(C^{\perp_L})_{Res} \cap ((C^{\perp_L})_{Res})^\perp \\ &=(C_{Res})^\perp \cap ((C_{Res})^\perp)^\perp \\
        &=(C_{Res})^\perp \cap ((C_{Res}) \\
        &=(LHull(C))_{Res},
        \end{align*}
        and 
        \begin{align*}
        (LHull(C^{\perp_L}))_{Tor}&= ((C^{\perp_L})_{Tor}) \cap  ((C^{\perp_L})_{Res})^\perp\\
        &=C_{Res}\cap (C_{Res})^\perp \\
        &=C_{Tor}\cap (C_{Res})^\perp \\
        &=(LHull(C))_{Tor}.
        \end{align*}
        Therefore, if $C$ is free, then $LHull(C)=LHull(C^{\perp_L})$.
        
        \item From Theorem \ref{Thm4e} (3), for an $E$-linear code $C$, we have 
        $$Hull(C^\perp)= C^\perp \cap (C^\perp)^\perp=C^\perp \cap C =Hull(C).$$

         \item Follows from Theorem \ref{Thm2}(2), for an $E$-linear code $C$, we have
         \begin{align*}
             (RHull(C^{\perp_R}))_{Res}&=(C^{\perp_R})_{Res} \cap ((C^{\perp_R})_{Tor})^\perp \\
             &=(C_{Tor})^\perp \cap (\mathbb{F}_2^n)^\perp \\
             &=(C_{Tor})^\perp \cap \{0\} \\
             &=\{0\},
        \end{align*}
        and $$(RHull(C^{\perp_R}))_{Tor}=(C^{\perp_R})_{Tor}=\mathbb{F}_2^n.$$
        Therefore, by (2) of Theorem \ref{Thm2}, $RHull(C)=RHull(C^{\perp_R})$ if $C_{Res} \cap (C_{Tor})^\perp=\{0\}$ and $C_{Tor}=\mathbb{F}_2^n$. 
        
        \item If $C$ is a free $E$-linear code, then (1) of Theorem \ref{Thm2} implies that
        \begin{align*}
        (LHull(C^{\perp_R}))_{Res}&= ((C^{\perp_R})_{Res})^\perp \cap (C^{\perp_R})_{Res}  \\
        &= (C_{Tor}) \cap (C_{Tor})^\perp  \\
        &=  (C_{Res}) \cap (C_{Tor})^\perp \\
        &=(RHull(C))_{Res},
        \end{align*}
        and
        \begin{align*}
        (LHull(C^{\perp_R}))_{Tor}&=((C^{\perp_R})_{Res})^\perp \cap ((C^{\perp_R})_{Tor}) \\ &=((C_{Tor})^\perp)^\perp \cap \mathbb{F}_2^n \\
        &=(C_{Tor}) \\
        &=(RHull(C))_{Tor}.
        \end{align*}
        This shows that $RHull(C)=LHull(C^{\perp_R})$. %Hence, the result.

        \item If $C$ is a free $E$-linear code, then (2) of Theorem \ref{Thm2} implies that
        \begin{align*}
        (RHull(C^{\perp_L}))_{Res}&= ((C^{\perp_L})_{Tor})^\perp \cap (C^{\perp_L})_{Res} \\
        &=  ((C_{Res})^\perp)^\perp \cap (C_{Res})^\perp \\
        &= C_{Res} \cap (C_{Res})^\perp \\
        &=(LHull(C))_{Res},
        \end{align*}
        and
      \begin{align*}
        (RHull(C^{\perp_L}))_{Tor} &= ((C^{\perp_L})_{Tor}) \\ &=(C_{Res})^\perp \\
        &=(C_{Res}) \cap (C_{Res})^\perp~~~~\hspace{1cm}( \because~ C_{Res} ~~\text{is SD })  \\
        &= (C_{Res})^\perp \cap C_{Tor}~~~~\hspace{1.3cm}( \because~ C_{Res}=C_{Tor}) \\
        &=(LHull(C))_{Tor}.
        \end{align*}
        Since $(RHull(C^{\perp_L}))_{Res}=(LHull(C))_{Res}$ and $(RHull(C^{\perp_L}))_{Tor}=(LHull(C))_{Tor}$, $LHull(C)=RHull(C^{\perp_L})$.
    \end{enumerate}
\end{proof}

Next, we have the following result, which investigates the freeness of the hull codes of a free $E$-linear code.

\begin{theorem}\label{Thm3a}
 For a free $E$-linear code $C$, its left hull and two-sided hull codes are also free.   
\end{theorem}
\begin{proof}
    If $C$ is a free $E$-linear code, then by (1) and (3) of Theorem \ref{Thm2}, we have

$$(LHull(C))_{Res}=C_{Res}\cap (C_{Res})^\perp,~~(LHull(C))_{Tor}=(C_{Res})^\perp \cap C_{Tor}$$
    and  $$(Hull(C))_{Res}=C_{Res} \cap (C_{Tor})^\perp,~~(Hull(C))_{Tor}=(C_{Res})^\perp \cap C_{Tor}.$$
Since $C$ is free, $C_{Res}=C_{Tor}$. Therefore,
    $$(LHull(C))_{Res}=(LHull(C))_{Tor}~~\text{and}~~ (Hull(C))_{Res}=(Hull(C))_{Tor}.$$
 Thus, the left hull and two-sided hull codes of the free $E$-code $C$ are also free.
 \end{proof}

Next, we include an example showing that the right hull code of a free $E$-linear code need not be free.

\begin{example}
  Consider the $E$-linear code $C$ with generator matrix 
  $$G=\begin{pmatrix}
      \kappa & 0 & 0& 0\\
      0 & \kappa & 0 & 0
  \end{pmatrix}.$$
  The residue and torsion codes of the code $C$ have the same generator matrix 
  $$G_{Res}=\begin{pmatrix}
      1 & 0 & 0& 0\\
      0 & 1 & 0 & 0
  \end{pmatrix}.$$
  Therefore, $C$ is free. Next, the dual of the residue code of $C$ is generated by the matrix
  $$H_{Res}=\begin{pmatrix}
      0 & 0 & 1& 0\\
      0 & 0 & 0 & 1
  \end{pmatrix}.$$
  Clearly, $C_{Res} \cap (C_{Res})^\perp=\{0\}$. Therefore, by (2) of Theorem \ref{Thm2}, we have $$(RHull(C))_{Res}=C_{Res} \cap (C_{Tor})^\perp=C_{Res} \cap (C_{Res})^\perp=\{0\} \neq (RHull(C))_{Tor}.$$
  Thus, the right hull $RHull(C)$ is not free.
\end{example}

 The following result plays a crucial role in our further investigations on hulls.
\begin{theorem}
    For a free $E$-linear code $C$, its left hull and two-sided hull codes coincide.
\end{theorem}
\begin{proof}
 Since $C$ is free, by (1) and (3) of Theorem \ref{Thm2}, we have 

$$(LHull(C))_{Res}=C_{Res}\cap (C_{Res})^\perp,~~(LHull(C))_{Tor}=(C_{Res})^\perp \cap C_{Tor}$$
    and  $$(Hull(C))_{Res}=C_{Res} \cap (C_{Tor})^\perp,~~(Hull(C))_{Tor}=(C_{Res})^\perp \cap C_{Tor}.$$
Further, the freeness of the $E$-linear code $C$ implies that $C_{Res}=C_{Tor}$. Therefore,
    $$(LHull(C))_{Res}=(Hull(C))_{Res}~~\text{and}~~ (LHull(C))_{Tor}=(Hull(C))_{Tor}.$$
 Thus, the left hull and two-sided hull codes of the free $E$-code $C$ are equal.
\end{proof}

According to the above results, the left hull and two-sided hull codes of a free $E$-linear code coincide, without assurance of the freeness of the right hull code. Therefore, from now onward, we focus exclusively on the two-sided hull code of a free $E$-linear code.

For a free $E$-linear code $C$, the following result provides us the generator matrix of its hull code.

\begin{theorem}\label{thm1a}
 If $C$ is a free $E$-linear code with $G$ as a generator matrix of the hull code of $C_{Res}$, then its hull code $Hull(C)$ has a generator matrix $\kappa G$.
 
\end{theorem}
\begin{proof}
  Since $C$ is a free $E$-code, $Hull(C)$ is also free by Theorem  \ref{Thm3a}. Next, by Corollary \ref{Cor1}, we have $(Hull(C))_{Res}=Hull(C_{Res})$. Moreover, from Theorem $1$ of \cite{Kushwaha2}, if $C$ is a free $E$-linear code and $G_1$ is a generator matrix of $C_{Res}$, then $\kappa G_1$ is a generator matrix of $C$. Therefore, if the hull code of $C_{Res}$ has a generator matrix $G$, then $\kappa G$ is a generator matrix of $Hull(C)$. 
\end{proof}

Now, we compute the hull-rank of a free $E$-linear code that is needed for the classification of optimal free $E$-linear codes.

\begin{theorem}\label{thm3}
   If $C$ is a free $E$-linear code, then
   $$rank(Hull(C))=dim(Hull(C_{Res})).$$
\end{theorem}
\begin{proof}
    Follows the definition of the rank of a free $E$-linear code and Theorem \ref{thm1a}.
\end{proof}

\section{Build-Up Constructions}
In this section, we give four build-up construction techniques that construct free $E$-linear codes with a larger length and hull-rank from free $E$-linear codes with a smaller length and hull-rank. Additionally, we give some examples of codes constructed by these build-up construction methods. Note that [$n,k$] represents a free $E$-linear code of length $n$ and rank $k$.\par
Now, we recall Proposition $1$ of \cite{Li18} for our further investigations.
%[ \cite{Li18}, Proposition $1$]
\begin{proposition}\label{Thm4} Let $C$ be an $\mathbb{F}_q$-linear $k$-dimensional code with hull-dimension $h$. If $G$ is a generator matrix of $C$ and $G^T$ denotes the transpose of $G$, then $h=k-rank(GG^T)$. 
\end{proposition}

The following result is our first build-up construction method.\par

\begin{theorem}[\textbf{Construction I}] Let $C$ be a free $E$-linear [$n,k$]-code. Also, let $G$ (with rows $r_1,r_2,\ldots,r_k$) and $H$ (with rows $s_1,s_2,\ldots,s_{n-k}$) be the generator and parity-check matrices of its residue code $C_{Res}$, respectively. Furthermore, fix a vector $u=(u_1, u_2, \ldots, u_n) \in \mathbb{F}_2^n$ with $\langle u, u \rangle =1$, and define the following scalars:

$$v_i=\langle u, r_i \rangle~~ \text{for}~~ 1 \leq i \leq k,  ~~\text{and}~~ w_j=\langle u, s_j \rangle ~~\text{for}~~ 1 \leq j \leq n-k.$$ 
Now, if the hull-rank of $C$ is $l$, then 
\begin{enumerate}
    \item[(a)] The $E$-linear code $D$ with generator matrix $G'$ given below is a free [$n+2, k+1$]-code with hull-rank $l+1$:

    \[
G' =
\left(
\begin{array}{cc|ccccc}
\kappa & 0 & \kappa u_1 & \kappa u_2 &   \cdots \cdots & & \kappa u_n \\ \hline
\kappa v_1 & \kappa v_1 & & &  \kappa r_1 & & \\
\kappa v_2 & \kappa v_2 & & &  \kappa r_2 & & \\
\vdots & \vdots & & &  \vdots & & \\
\kappa v_k & \kappa v_k & & &  \kappa r_k & &
\end{array}
\right)
.\]

    \item[(b)] The code $D$ has the following parity-check matrix:

\[
H' =
\left(
\begin{array}{cc|ccccc}
\kappa & 0 & \kappa u_1 & \kappa u_2 &   \cdots \cdots & & \kappa u_n \\ \hline
\kappa w_1 & \kappa w_1 & & &  \kappa s_1 & & \\
\kappa w_2 & \kappa w_2 & & &  \kappa s_2 & & \\
\vdots & \vdots & & &  \vdots & & \\
\kappa w_{n-k} & \kappa w_{n-k} & & &  \kappa s_{n-k} & &
\end{array}
\right)
.\]
    
\end{enumerate}
\end{theorem}
\begin{proof}

\begin{enumerate}
    \item[(a)]
    By the form of $G'$, it is clear that it generates a free $E$-linear code, and its residue code has the generator matrix
    \[
G_1 =
\left(
\begin{array}{cc|ccccc}
1 & 0 &  u_1 & u_2 &   \cdots \cdots & & u_n \\ \hline
 v_1 &  v_1 & & &   r_1 & & \\
 v_2 &  v_2 & & &   r_2 & & \\
\vdots & \vdots & & &  \vdots & & \\
 v_k & v_k & & &   r_k & &
\end{array}
\right)
.\]
Then, we have 
\[
G_1G_1^T=
\left(
\begin{array}{ccccccc}
0 & 0 &  0 & 0 &   \cdots \cdots & & 0 \\  
 0 &   & & &   & & \\
 0 &   & & GG^T&  & & \\
\vdots &  & & & & & \\
0 &  & & &   & &
\end{array}
\right)
.\]
Therefore, $rank(G_1G_1^T)=rank(GG^T)$. Since $C$ is free, by Theorem \ref{thm3}, we have $dim(Hull(C_{Res}))=rank(Hull(C))=l$.  Also, by Proposition \ref{Thm4}, $rank(GG^T)=k-l$ and hence $rank(G_1G_1^T)=k-l$. Again, Proposition \ref{Thm4} implies that the hull-dimension of $D_{Res}$ is given by $(k+1)-rank(G_1G_1^T)=(k+1)-(k-l)=l+1$. Since, for a free $E$-linear code $D$, $rank(D)=dim(D_{Res})$, we have $rank(Hull(D))=dim(Hull(D_{Res}))=l+1$. This completes the proof of the first part.

\item[(b)] We have  
\[
\pi(H') =H_1=
\left(
\begin{array}{cc|ccccc}
1 & 0 &  u_1 &  u_2 &   \cdots \cdots & & u_n \\ \hline
 w_1 & w_1 & & &  s_1 & & \\
 w_2 &  w_2 & & &   s_2 & & \\
\vdots & \vdots & & &  \vdots & & \\
 w_{n-k} & w_{n-k} & & &  s_{n-k} & &
\end{array}
\right)
.\]
Then, we have 
\[
G_1H_1^T=
\left(
\begin{array}{ccccccc}
0 & 0 &  0 & 0 &   \cdots \cdots & & 0 \\  
 0 &   & & &   & & \\
 0 &   & & GH^T&  & & \\
\vdots &  & & & & & \\
0 &  & & &   & &
\end{array}
\right)=\boldsymbol{0}
.\]

The dimension of the code generated by $H_1$ is $(n-k)+1$, since the first row of $H_1$ is not a linear combination of the other rows of $H_1$. Further, 
$$dim \langle H_1 \rangle =(n-k)+1=(n+2)-(k+1)=dim((D^\perp)_{Res}).$$
Therefore, $H_1$ is a parity-check matrix of $(D)_{Res}$. Thus, $H'$ is a parity-check matrix of $D$.

\end{enumerate}
\end{proof}

Next, we give an example which utilizes a free $E$-linear [$6,4$]-code with hull-rank $2$ to construct a free [$8,5$]-code over $E$ with hull-rank $3$ by \textbf{Construction I}.

\begin{example}
    Consider the free $E$-linear code $C$ with generator  matrix 
    \[
    N=
\left (
\begin{array}{cccccc}
    \kappa & 0 & 0 & 0 & \kappa & 0  \\
0 & \kappa & 0 & 0 & \kappa & \kappa \\
0 & 0 & \kappa & 0 & 0 & \kappa \\
0 & 0 & 0 & \kappa & \kappa & \kappa 
\end{array}
  \right ) 
    .\]
    Clearly, $C$ is a [$6,4$]-code with hull-rank $2$, and  $C_{Res}$ has a generator matrix
    
\[
    G=
\left (
\begin{array}{cccccc}
    1 & 0 & 0 & 0 & 1 & 0  \\
0 & 1 & 0 & 0 & 1 & 1 \\
0 & 0 & 1 & 0 & 0 & 1 \\
0 & 0 & 0 & 1 & 1 & 1 
\end{array}
  \right ) 
    .\]
     Next, if we take $u=(1,0,0,1,0,1)$, then $\langle u, u \rangle =1$ and $v_1=1, v_2=1, v_3=1$ and $v_4=0$. Now, consider the code $D$ with the following generator matrix

     \[
    G'=
\left (
\begin{array}{cc|cccccc}
\kappa & 0 & \kappa & 0 & 0 & \kappa & 0 & \kappa \\ \hline
\kappa & \kappa&     & &  & &  &   \\ 
\kappa & \kappa&  &  & N &  &  &  \\
\kappa & \kappa&  &  &  &  &  &  \\
0 & 0 &  &  &  &  &  &  
\end{array}
  \right ) 
    .\]
    Then, by \textbf{Construction I}, $D$ is a free [$8,5$]-code over $E$ with hull-rank $3$.

\end{example}

Now, we have our second build-up construction method as follows. \par

\begin{theorem}[\textbf{Construction II}] Let $C$ be a free $E$-linear [$n,k$]-code. Let $G$ (with rows $r_1,r_2,\ldots,r_k$) and $H$ (with rows $s_1,s_2,\ldots,s_{n-k}$) be the generator and parity-check matrices of its residue code $C_{Res}$, respectively. Furthermore, fix a vector $u=(u_1, u_2, \ldots, u_n) \in \mathbb{F}_2^n$ with $\langle u, u \rangle =0$ and $\langle u, r_i \rangle=0~~ \text{for}~~ 1 \leq i \leq k$. Also, define
$w_j=\langle u, s_j \rangle ~~\text{for}~~ 1 \leq j \leq n-k$. With these assumptions, if the hull-rank of $C$ is $l$, then

\begin{enumerate}
    \item[(a)] The code $D$ with the generator matrix

    \[
G' =
\left(
\begin{array}{cc|ccccc}
\kappa & \kappa & \kappa u_1 & \kappa u_2 &   \cdots \cdots & & \kappa u_n \\ \hline
0 & 0 & & &  \kappa r_1 & & \\
0 & 0 & & &  \kappa r_2 & & \\
\vdots & \vdots & & &  \vdots & & \\
0 & 0 & & &  \kappa r_k & &
\end{array}
\right)
\]

    is a free $E$-linear [$n+2, k+1$]-code with hull-rank $l+1$.
    \item[(b)] The code $D$ has the following parity-check matrix:

\[
H' =
\left(
\begin{array}{cc|ccccc}
\kappa & \kappa & \kappa u_1 & \kappa u_2 &   \cdots \cdots & & \kappa u_n \\ \hline
0 & \kappa w_1 & & &  \kappa s_1 & & \\
0 & \kappa w_2 & & &  \kappa s_2 & & \\
\vdots & \vdots & & &  \vdots & & \\
0 & \kappa w_{n-k} & & &  \kappa s_{n-k} & &
\end{array}
\right)
.\]
    
\end{enumerate}
\end{theorem}

\begin{proof}

\begin{enumerate}
    \item[(a)]
     It is clear that $G'$ generates a free $E$-linear code, and its residue code has the generator matrix
   \[
G_1 =
\left(
\begin{array}{cc|ccccc}
1 & 1 &  u_1 &  u_2 &   \cdots \cdots & & u_n \\ \hline
0 & 0 & & &  r_1 & & \\
0 & 0 & & &   r_2 & & \\
\vdots & \vdots & & &  \vdots & & \\
0 & 0 & & & r_k & &
\end{array}
\right)
\]

Then, we have 
\[
G_1G_1^T=
\left(
\begin{array}{ccccccc}
0 & 0 &  0 & 0 &   \cdots \cdots & & 0 \\  
 0 &   & & &   & & \\
 0 &   & & GG^T&  & & \\
\vdots &  & & & & & \\
0 &  & & &   & &
\end{array}
\right)
.\]
The proof follows similar arguments to part $(a)$ of \textbf{Construction I}.

\item[(b)] We have  
\[
\pi(H') =H_1=
\left(
\begin{array}{cc|ccccc}
1 & 1 &  u_1 &  u_2 &   \cdots \cdots & & u_n \\ \hline
0 & w_1 & & &  s_1 & & \\
 0 &  w_2 & & &   s_2 & & \\
\vdots & \vdots & & &  \vdots & & \\
0 & w_{n-k} & & &  s_{n-k} & &
\end{array}
\right)
.\]
Then 
\[
G_1H_1^T=
\left(
\begin{array}{ccccccc}
0 & 0 &  0 & 0 &   \cdots \cdots & & 0 \\  
 0 &   & & &   & & \\
 0 &   & & GH^T&  & & \\
\vdots &  & & & & & \\
0 &  & & &   & &
\end{array}
\right)=\boldsymbol{0}
.\]

The remainder of the proof is parallel to the part $(b)$ of \textbf{Construction I}.

\end{enumerate}
\end{proof}

The following example utilizes a free $E$-linear [$9,4$]-code with hull-rank $4$ to construct a free [$11,5$]-code over $E$ with hull-rank $5$ by \textbf{Construction II}.

\begin{example}
    Let $C$ be an $E$-linear code with generator matrix  

\[
    N=
\left (
\begin{array}{ccccccccc}
    \kappa & 0 & 0 & 0 & \kappa & \kappa & 0 & \kappa & 0  \\
0 & \kappa & 0 & 0 & 0 & \kappa & 0 & \kappa & \kappa\\
0 & 0 & \kappa & 0 & \kappa & \kappa & 0 & 0 & \kappa\\
0 & 0 & 0 &\kappa & \kappa & 0 & 0 & \kappa & \kappa 
\end{array}
  \right ) 
    .\]
   Then, $C_{Res}$ is generated by the matrix
    
\[
    G=
\left (
\begin{array}{ccccccccc}
    1 & 0 & 0 & 0 & 1 & 1 & 0 & 1 & 0 \\
0 & 1 & 0 & 0 & 0 & 1 & 0 & 1 & 1 \\
0 & 0 & 1 & 0 & 1 & 1 & 0 & 0 & 1 \\
0 & 0 & 0 & 1 & 1 & 0 & 0 & 1 & 1 
\end{array}
  \right ) 
    .\]

 Clearly, $C$ is a free [$9,4$]-code over $E$ with hull-rank $4$. If we take $u=(1,0,0,0,1,1,0,1,0)$, then $\langle u, u \rangle =0$ and $\langle u, r_i \rangle=0$ for $1 \leq i \leq 4$. Now, let $D$ be the $E$-linear code generated by the following matrix 

\[
    G'=
\left (
\begin{array}{cc|ccccccccc}
\kappa & \kappa & \kappa & 0 & 0 & 0 & \kappa & \kappa & 0 & \kappa & 0 \\ \hline 
 0 & 0 &    &  &  &  &  &  &  &  &   \\
0 & 0 &  &  &  & N &  &  &  &  & \\
0 & 0 &  &  &  &  &  &  &  &  & \\
0 & 0 &  &  &  &  &  &  &  &  &  
\end{array}
  \right ) 
    .\]
Then, by \textbf{Construction II}, $D$ is a free [$11,5$]-code over $E$ with hull-rank $5$.

\end{example}

The following result is our third build-up construction method.

\begin{theorem}[\textbf{Construction III}] Let $C$ be a free $E$-linear [$n,k$]-code. Also, let $G$ (with rows $r_1,r_2,\ldots,r_k$) and $H$ (with rows $s_1,s_2,\ldots,s_{n-k}$) be the generator and parity-check matrices of its residue code $C_{Res}$, respectively. Furthermore, fix a vector $u=(u_1, u_2, \ldots, u_n) \in \mathbb{F}_2^n$ with $\langle u, u \rangle =0$, and define the following scalars:

$$v_i=\langle u, r_i \rangle~~ \text{for}~~ 1 \leq i \leq k,  ~~\text{and}~~ w_j=\langle u, s_j \rangle ~~\text{for}~~ 1 \leq j \leq n-k.$$ 
Now, assume that not all $v_i$'s are zero and hull-rank of $C$ is $l$, then

\begin{enumerate}
    \item[(a)] The $E$-linear code $D$ with generator matrix $G'$ given below is a free [$n+2, k+1$]-code with hull-rank $l,l+1$ or $l+2$:

    \[
G' =
\left(
\begin{array}{cc|ccccc}
\kappa & \kappa & \kappa u_1 & \kappa u_2 &   \cdots \cdots & & \kappa u_n \\ \hline
\kappa v_1 & 0 & & &  \kappa r_1 & & \\
\kappa v_1 & 0 & & &  \kappa r_2 & & \\
\vdots & \vdots & & &  \vdots & & \\
\kappa v_1 & 0 & & &  \kappa r_k & &
\end{array}
\right)
.\]

    \item[(b)] The code $D$ has the following parity-check matrix:

\[
H' =
\left(
\begin{array}{cc|ccccc}
\kappa & \kappa & \kappa u_1 & \kappa u_2 &   \cdots \cdots & & \kappa u_n \\ \hline
0 & \kappa w_1 & & &  \kappa s_1 & & \\
0 & \kappa w_2 & & &  \kappa s_2 & & \\
\vdots & \vdots & & &  \vdots & & \\
0 & \kappa w_{n-k} & & &  \kappa s_{n-k} & &
\end{array}
\right)
.\]
    
\end{enumerate}
\end{theorem}
\begin{proof}

\begin{enumerate}
    \item[(a)]
    Follows the form of $G'$, it generates a free $E$-linear code, and its residue code has the generator matrix
   \[
G_1 =
\left(
\begin{array}{cc|ccccc}
1 & 1 &  u_1 &  u_2 &   \cdots \cdots & & u_n \\ \hline
v_1 & 0 & & &  r_1 & & \\
v_1 & 0 & & &   r_2 & & \\
\vdots & \vdots & & &  \vdots & & \\
v_1 & 0 & & & r_k & &
\end{array}
\right)
.\]

Then, applying elementary row operations on $G_1$, we get  
\[
G'_1=
\left(
\begin{array}{cc|ccccc}
1 & 1 &  u_1 & u_2 &   \cdots \cdots & & u_n \\ \hline 
 1 & 0  & & &   & & \\
 0 &  0 & & & G_2  & & \\
\vdots & \vdots  & & & & & \\
0 & 0 & & &   & &
\end{array}
\right)
,\] 
where $\langle G_2 \rangle =\langle G \rangle $. Now, we have 
\[
G'_{1} (G'_{1})^T=
\left(
\begin{array}{ccccccc}
0 & 0 &  0 & 0 &   \cdots \cdots & & 0 \\  
 0 &   & & &   & & \\
 0 &   & & M &  & & \\
\vdots &  & & & & & \\
0 &  & & &   & &
\end{array}
\right)
,\]
where $M=(1,0,\ldots, 0)^T(1,0,\ldots,0)+G_2G_2^T$. We know that if $M_1$ and $M_2$ are two matrices, then $rank(M_1+M_2)\leq rank(M_1)+rank(M_2)$. Therefore,
\begin{align*}
rank(M) & \leq rank((1,0,\ldots, 0)^T(1,0,\ldots,0))+ rank(G_2G_2^T) \\
&\leq 1+rank(G_2G_2^T).
\end{align*}
We observe that $(1,0,\ldots, 0)^T(1,0,\ldots,0)$ affects exclusively the first row of $G_2G_2^T$. Hence, $rank(M)$ can be $rank(G_2G_2^T)$, $rank(G_2G_2^T)-1$ or $rank(G_2G_2^T)+1$. Then, the following cases arise:\\
\justifying \textbf{Case 1:} Let $rank(M)=rank(G_2G_2^T)$. Then
\begin{align*}
    rank(G_1G_1^T) &=rank(G'_1(G'_{1})^T) \\
    &= rank(M) \\
    &= rank(G_2G_2^T) \\
    &=rank(GG^T) \\
    &=k-l
.\end{align*}
Therefore, hull-dimension of $D_{Res}$ is $(k+1)-(k-l)=l+1$. Thus, by Theorem \ref{thm3}, the hull-rank of the code $D$ is $l+1$ in this case.\\
\justifying \textbf{Case 2:} Let $rank(M)=rank(G_2G_2^T)-1$. Then
\begin{align*}
    rank(G_1G_1^T) &=rank(G'_1(G'_{1})^T) \\
    &= rank(M)= rank(G_2G_2^T) -1 \\
    &=rank(GG^T) -1 \\
    &=k-l-1.
    \end{align*}
Therefore, hull-dimension of $D_{Res}$ is $(k+1)-(k-l-1)=l+2$. Thus, by Theorem \ref{thm3}, the hull-rank of the code $D$ is $l+2$ in this case.\\
\justifying \textbf{Case 3:} Let $rank(M)=rank(G_2G_2^T)+1$. Then
\begin{align*}
    rank(G_1G_1^T) &=rank(G'_1(G'_{1})^T)\\
    &= rank(M) \\
    &= rank(G_2G_2^T) +1 \\
    &=rank(GG^T) +1 \\
    &=k-l+1.
\end{align*}
Therefore, $D_{Res}$ has hull-dimension $(k+1)-(k-l+1)=l$. Hence, by Theorem \ref{thm3}, the hull-rank of the code $D$ is $l$ in this case.

\item[(b)] This part follows similar approach as the part $(b)$ of \textbf{Construction I}.

\end{enumerate}
\end{proof}

Now, we have a supporting example for a free $E$-linear [$10,6$]-code with hull-rank $1$ to construct a free [$12,7$]-code over $E$ with hull-rank $3$ by \textbf{Construction III}.

\begin{example}
    Let $C$ be a  free $E$-linear [$10,6$]-code with hull-rank $1$, and $C_{Res}$ has a generator matrix 
\[
G =
\left(
\begin{array}{cccccccccc}
1 & 0 & 0 & 0 &   0 & 0& 1 & 0 & 0 & 1 \\
0 & 1 &0 &0 &  0 & 0 & 0 & 1 & 0 & 1 \\
0 & 0 & 1&0 &0 &0& 1 & 1 & 0 & 1 \\
0 & 0 &0 &1 &  0 &0& 1 & 1 & 1 & 1 \\
0 & 0 &0 &0& 1&  0 & 1 &1& 0 &0 \\
0 & 0 &0 &0& 0&  1&  0 & 0 & 1 &1
\end{array}
\right)
.\]
  If $u=(1,1,1,1,1,1,1,1,1,1)$, then $\langle u, u \rangle =0$, and $v_1=v_2=v_4=v_5=v_6=1, v_3=0$. Now, consider the code $D$ with generator matrix   
  \[
G' =
\left(
\begin{array}{cc|cccccccccc}
\kappa & \kappa & \kappa & \kappa &\kappa &\kappa &\kappa &\kappa & \kappa &\kappa &\kappa &\kappa \\ \hline
\kappa & 0 &  &  &  & &    &&  & &   &  \\
\kappa & 0 & &  & & &   &  & &  &  &  \\
0 & 0 &  &  & &\kappa G & &&  &  &  &   \\
\kappa & 0 & &  & &&  && &  &  &  \\
\kappa & 0 & & & && &   & &  & &  \\
\kappa & 0 &  &  & && &  &   & &  & \ 
\end{array}
\right)
.\]
Then, by \textbf{Construction III}, $D$ is a free [$12,7$]-code over $E$ with hull-rank $3$.

\end{example}

The following result is our fourth build-up construction method.

\begin{theorem}[\textbf{Construction IV}] Let $C$ be a free $E$-linear [$n,k$]-code. Also, let $G$ (with rows $r_1,r_2,\ldots,r_k$) and $H$ (with rows $s_1,s_2,\ldots,s_{n-k}$) be the generator and parity-check matrices of its residue code $C_{Res}$, respectively. Furthermore, fix a vector $u=(u_1, u_2, \ldots, u_n) \in \mathbb{F}_2^n$ with $\langle u, u \rangle =0$, and define the following scalars:

$$v_i=\langle u, r_i \rangle~~ \text{for}~~ 1 \leq i \leq k,  ~~\text{and}~~ w_j=\langle u, s_j \rangle ~~\text{for}~~ 1 \leq j \leq n-k.$$ 
Assume the hull-rank of $C$ is $l$, then 
 
\begin{enumerate}
    \item[(a)] The code $D$ with the generator matrix

    \[
G' =
\left(
\begin{array}{cc|ccccc}
\kappa & 0 & \kappa u_1 & \kappa u_2 &   \cdots \cdots & & \kappa u_n \\ \hline
\kappa v_1 & \kappa v_1 & & &  \kappa r_1 & & \\
\kappa v_2 & \kappa v_2 & & &  \kappa r_2 & & \\
\vdots & \vdots & & &  \vdots & & \\
\kappa v_k & \kappa v_k & & &  \kappa r_k & &
\end{array}
\right)
\]

    is a free $E$-linear [$n+2, k+1$]-code with hull-rank $l$.
    \item[(b)] The code D has the following parity-check matrix: 

\[
H' =
\left(
\begin{array}{cc|ccccc}
0 & \kappa & \kappa u_1 & \kappa u_2 &   \cdots \cdots & & \kappa u_n \\ \hline
\kappa w_1 & \kappa w_1 & & &  \kappa s_1 & & \\
\kappa w_2 & \kappa w_2 & & &  \kappa s_2 & & \\
\vdots & \vdots & & &  \vdots & & \\
\kappa w_{n-k} & \kappa w_{n-k} & & &  \kappa s_{n-k} & &
\end{array}
\right)
.\]
    
\end{enumerate}
\end{theorem}
\begin{proof}

\begin{enumerate}
    \item[(a)]
    It is clear that $G'$ generates a free $E$-linear code, and its residue code has the generator matrix
    \[
G_1 =
\left(
\begin{array}{cc|ccccc}
1 & 0 &  u_1 & u_2 &   \cdots \cdots & & u_n \\ \hline
 v_1 &  v_1 & & &   r_1 & & \\
 v_2 &  v_2 & & &   r_2 & & \\
\vdots & \vdots & & &  \vdots & & \\
 v_k & v_k & & &   r_k & &
\end{array}
\right)
.\]
Then
\[
G_1G_1^T=
\left(
\begin{array}{ccccccc}
1 & 0 &  0 & 0 &   \cdots \cdots & & 0 \\  
 0 &   & & &   & & \\
 0 &   & & GG^T&  & & \\
\vdots &  & & & & & \\
0 &  & & &   & &
\end{array}
\right)
.\]
Hence, $rank(G_1G_1^T)=rank(GG^T)+1=k-l+1$. Again, Proposition \ref{Thm4} implies that the hull-dimension of $D_{Res}$ is given by $(k+1)-rank(G_1G_1^T)=(k+1)-(k-l+1)=l$. Therefore, $rank(Hull(D))=dim(Hull(D_{Res}))=l$.

\item[(b)] This part can be proved similar to the part $(b)$ of previous constructions.

\end{enumerate}
\end{proof}

This section concludes with the following example, which illustrates \textbf{Construction IV}.
\begin{example}
  Let $C$ be a free [$10,5$]-code over $E$ with hull-rank $5$, and $C_{Res}$ has a generator matrix 
\[
G =
\left(
\begin{array}{cccccccccc}
1 & 0 & 0 & 0 &   0 & 0& 0 & 1 & 0 & 0\\
0 & 1 &0 &0 &  0 & 0 & 1&0 & 0 & 0 \\
0 & 0 & 1&0 &0 &0& 0 & 0 & 0 & 1 \\
0 & 0 &0 &1 &  0 &0& 0 & 0 & 1 & 0 \\
0 & 0 &0 &0& 1&  1 & 0 &0& 0 &0 
\end{array}
\right)
.\]
  If we take $u=(1,0,1,1,0,1,0,0,0,0)$, then $\langle u, u \rangle =0$, and $v_1=v_3=v_4=v_5=v_6=1, v_2=0$. Now, consider the code $D$ with generator matrix   
  \[
G' =
\left(
\begin{array}{cc|cccccccccc}
\kappa & 0 & \kappa & 0 &\kappa &\kappa &0 &\kappa & 0 &0 &0 &0 \\ \hline
\kappa & \kappa &  &  &  &  &    & &  &  & &    \\
0 & 0 &  &  & & &   &  & & &  &  \\
\kappa & \kappa &  & & & &\kappa G & &  &  &  &   \\
\kappa & \kappa & &  & & &   & &  &  &  &  \\
\kappa & \kappa & &  & & & &   &  &  & &  \\

\end{array}
\right)
.\]
Then, by \textbf{Construction IV}, $D$ is a free [$12,6$]-code over $E$ with hull-rank $5$.

\end{example}

\section{Classification}
This section investigates the permutation equivalence of free $E$-linear codes and classifies the optimal free $E$-linear codes with a fixed hull-rank. We also study the hull-variation problem for the free $E$-linear codes.

\begin{definition}[Permutation-equivalent codes] A code $C$ is called permutation-equivalent to a code $D$ if $C$ can be obtained from $D$ by a suitable coordinate permutation.
 \end{definition}
 
The following result from \cite{Alah23} investigates the permutation equivalence of two free $E$-linear codes. 

\begin{theorem}( \cite{Alah23}, Theorem $16$)\label{Thm5}
  Two free $E$-linear codes $C$ and $D$ are permutation-equivalent if and only if $C_{Res}$ and $D_{Res}$  are permutation-equivalent. 
\end{theorem}

To study the hull-variation problem for the free $E$-linear codes, we have the following.
\begin{theorem}\label{Thm6}
    Let $C$ be a free $E$-linear [$n,k$]-code and $\alpha$ be a permutation. Then 
    \begin{enumerate}
        \item $\langle \alpha(w), \alpha(z) \rangle=\langle w, z \rangle$ for all $w,z \in E^n$,
        \item $\alpha(C^\perp)=\alpha(C)^\perp$.
    \end{enumerate}
\end{theorem}
\begin{proof}
   \begin{enumerate}
       \item Let $w=(w_1,w_2, \ldots, w_n), z=(z_1,z_2, \ldots, z_n) \in E^n$ and $\alpha$ be a permutation. Then
       $$\alpha(w)=(w_{\alpha(1)}, w_{\alpha(2)}, \ldots, w_{\alpha(n)})~~\text{and}~~~ \alpha(z)=(z_{\alpha(1)}, z_{\alpha(2)}, \ldots, z_{\alpha(n)}).$$
       Therefore, 
       $$\langle \alpha(w), \alpha(z) \rangle=\sum_{i=1}^n w_{\alpha(i)}z_{\alpha(i)}=\sum_{i=1}^n w_iz_i=\langle w, z \rangle .$$
       \item Let $w \in C^\perp$ and $\alpha$ be a permutation. Then $\alpha(w) \in \alpha(C^\perp)$. Now, for any $z\in \alpha(C)$, there exists a codeword $x \in C$ such that $\alpha(x)=z$. Hence, by the first part, we have 
       $$\langle z, \alpha(w) \rangle=\langle \alpha(x), \alpha(w) \rangle = \langle x,w  \rangle = 0.$$
       Therefore, $\alpha(w) \in \alpha(C)^\perp$, and hence $\alpha(C^\perp) \subseteq (\alpha(C))^\perp$. Further, for the free [$n,k$]-code $C$ over $E$, 
       $$rank(\alpha(C^\perp))=rank(C^\perp)=n-rank(C)=n-rank(\alpha(C))=rank(\alpha(C)^\perp).$$
       This implies that $\alpha(C^\perp)=\alpha(C)^\perp$.
   \end{enumerate} 
\end{proof}

We now investigate the hull-variation problem for the free $E$-linear codes.

\begin{theorem}
    The hull-ranks of any two permutation-equivalent free $E$-linear codes are identical.
\end{theorem}
\begin{proof}
  Let $C$ be a free $E$-linear code and $\alpha$ be an arbitrary permutation. Then, Theorem \ref{Thm6} implies that 
$$\alpha(Hull(C))=\alpha(C \cap C^\perp)= \alpha(C) \cap \alpha(C^\perp)= \alpha(C) \cap \alpha(C)^\perp = Hull(\alpha(C)).$$
Therefore, $$rank(Hull(\alpha(C)))=rank(\alpha(Hull(C)))=rank(Hull(C)).$$
  %This completes the proof.
\end{proof}

%In what follows, we utilize the notion of optimal codes in a way consistent with the approach in \cite{Li}.
\begin{definition}[$l_i$-optimal free code]\label{defn2}
    A free $E$-linear [$n,k$]-code $C$ with hull-rank $l=i$ is called an $l_i$-optimal code if the minimum distance of the code $C$ is maximum among all the free [$n,k$]-codes over $E$ with hull-rank $i$.
\end{definition}

Here, we classify optimal free $E$-linear codes with a fixed hull-rank. Since LCD $E$-codes are free $E$-linear codes whose hull has rank $0$, and their classification is already provided in \cite{Kushwaha2}. Therefore, we classify only the free $E$-linear codes with non-zero hull-ranks. We rely on \cite{Li} for binary linear codes with a fixed hull-dimension and use Theorem $1$ of \cite{Kushwaha2} to get all the free $E$-linear codes, and then use Theorem \ref{Thm5} to find all the permutation-inequivalent free linear codes over $E$. Further, Theorem \ref{thm3} helps to find their hull-ranks. Later, we use the Definition \ref{defn2} to determine their optimality and list them in Tables \ref{Tab2}, \ref{Tab3}, \ref{Tab4}, \ref{Tab5}, \ref{Tab6}, \ref{Tab7}, \ref{Tab8a}, \ref{Tab8}, \ref{Tab9}, \ref{Tab9a}, \ref{Tab10}, \ref{Tab10a} and \ref{Tab11}. In the tables, $I_m$ is the identity matrix of order $m$. Here, all computations have been carried out using MAGMA \cite{Magma}.

\begin{table}[ht]
\centering
\begin{minipage}{0.4\textwidth}
\centering
\caption{\label{Tab2}$l_1$-optimal free codes (Part I).}
\begin{tabular}{|c|c|}
\hline
Generator Matrix & $[n,k,d]$ \\ \hline
$\begin{pmatrix}
    \kappa & \kappa
    \end{pmatrix}$ & $[2, 1, 2]$  \\ 
$\begin{pmatrix}
\kappa & \kappa & 0
\end{pmatrix}$ & $[3, 1, 2]$ \\ 
$\begin{pmatrix} 
\kappa & 0 & \kappa \\ 
0 & \kappa & 0 
\end{pmatrix}$ & $[3, 2, 1]$ \\ 
$\begin{pmatrix}
\kappa & \kappa & \kappa
\end{pmatrix}$ & $[4, 1, 4]$ \\ 
$\begin{pmatrix} 
\kappa & 0 & 0 & 0 \\ 
0 & \kappa & \kappa & 0 
\end{pmatrix}$ & $[4, 2, 1]$ \\ 
$\begin{pmatrix} 
 & \kappa \\
 \kappa I_3  & \kappa \\ 
   & \kappa 
\end{pmatrix}$ & $[4, 3, 2]$  \\ 
$\begin{pmatrix} 
\kappa & \kappa & \kappa & 0 & \kappa 
\end{pmatrix}$ & $[5, 1, 4]$  \\ 
$\begin{pmatrix} 
\kappa & 0 & 0 & \kappa & \kappa \\ 
0 & \kappa & \kappa & \kappa & 0 
\end{pmatrix}$ & $[5, 2, 3]$ \\

$\begin{pmatrix} 

 & 0 & \kappa  \\
 \kappa I_3  & \kappa & \kappa \\ 
   & \kappa & \kappa
\end{pmatrix}$ & $[5, 3, 2]$ \\ 

$\begin{pmatrix} 
 & 0 & \kappa  \\
 \kappa I_3  & \kappa & 0\\ 
   & 0 & \kappa
\end{pmatrix}$ & $[5, 3, 2]$ \\ 

$\begin{pmatrix} 
 & 0 & \kappa  \\
 \kappa I_3  & 0 & \kappa\\ 
   & 0 & \kappa
\end{pmatrix}$ & $[5, 3, 2]$  \\ 

$\begin{pmatrix} 
   & \kappa  \\
 \kappa I_4    & 0\\ 
  & 0 \\
 & 0
\end{pmatrix}$ & $[5, 4, 1]$  \\ 

$\begin{pmatrix} 
   & \kappa  \\
 \kappa I_4    & \kappa\\ 
     & \kappa \\
 & 0
\end{pmatrix}$ & $[5, 4, 1]$ \\ 

$\begin{pmatrix} 
 \kappa & \kappa & \kappa & \kappa    & \kappa
    & \kappa 
\end{pmatrix}$ & $[6, 1, 6]$ \\ 

$\begin{pmatrix} 
 \kappa & 0 & \kappa & 0    & 0
    & \kappa \\
     0 & \kappa & \kappa & \kappa    & 0
    & 0
\end{pmatrix}$ & $[6, 2, 3]$ \\ 

$\begin{pmatrix} 
   & 0 & 0 &\kappa  \\
 \kappa I_3  & 0 &0  & \kappa\\ 
   &0 & 0 & \kappa 
\end{pmatrix}$ & $[6, 3, 2]$ \\ 

$\begin{pmatrix} 
   & \kappa & 0 &\kappa  \\
 \kappa I_3  & 0 &0  & \kappa\\ 
     &0 & \kappa & \kappa 
\end{pmatrix}$ & $[6, 3, 2]$ \\ 

$\begin{pmatrix} 
   & 0 & \kappa &\kappa  \\
 \kappa I_3  & 0 &0  & \kappa\\ 
     &0 & \kappa & \kappa 
\end{pmatrix}$ & $[6, 3, 2]$ \\ 

$\begin{pmatrix} 
   & \kappa & 0 &0  \\
 \kappa I_3  & 0 &0  & \kappa\\ 
     &0 & 0 & \kappa 
\end{pmatrix}$ & $[6, 3, 2]$ \\ 

$\begin{pmatrix} 
   & \kappa & \kappa &\kappa  \\
 \kappa I_3  & 0 &0  & \kappa\\ 
     &0 & 0 & \kappa 
\end{pmatrix}$ & $[6, 3, 2]$\\

 \hline
\end{tabular}
\end{minipage}
\hspace{0.05\textwidth}
\begin{minipage}{0.45\textwidth}
\centering
\caption{\label{Tab3} $l_1$-optimal free codes (Part II).}
\begin{tabular}{|c|c|}
\hline

Generator Matrix & [$n,k,d$] \\
\hline

$\begin{pmatrix} 
   & \kappa & 0 &0  \\
 \kappa I_3  & 0 &0  & \kappa\\ 
     &0 & \kappa & \kappa 
\end{pmatrix}$ & $[6, 3, 2]$ \\ 

$\begin{pmatrix} 
   & 0 & \kappa &0  \\
 \kappa I_3  & 0 &0  & \kappa\\ 
     &\kappa & \kappa & \kappa 
\end{pmatrix}$ & $[6, 3, 2]$ \\

$\begin{pmatrix}
     & 0 & \kappa \\
     \kappa I_4 & \kappa & 0 \\
     & \kappa & \kappa \\
     & 0 & 0 
 \end{pmatrix}$ & [$6,4,1$] \\

$\begin{pmatrix}
     & 0 & \kappa \\
     \kappa I_4 & 0 & \kappa \\
     & 0 & 0 \\
     & 0 & \kappa 
 \end{pmatrix}$ & [$6,4,1$] \\

 $\begin{pmatrix}
     & 0 & \kappa \\
     \kappa I_4 & 0 & 0 \\
     & 0 & 0 \\
     & 0 & 0 
 \end{pmatrix}$ & [$6,4,1$] \\
 $\begin{pmatrix}
     & 0 & \kappa \\
     \kappa I_4 & \kappa & 0 \\
     & \kappa & 0 \\
     & 0 & 0 
 \end{pmatrix}$ & [$6,4,1$] \\
 $\begin{pmatrix}
     & \kappa \\
     & \kappa \\
     \kappa I_5 & \kappa \\
     & \kappa \\
     & \kappa
 \end{pmatrix}$ & [$6,5,2$] \\

$\begin{pmatrix}
     \kappa & \kappa & \kappa & \kappa & 0 & \kappa & \kappa
 \end{pmatrix}$ & [$7,1,6$] \\
 $\begin{pmatrix}
     \kappa & 0 & \kappa &  \kappa & 0 & \kappa & \kappa \\
     0 & \kappa & 0 & \kappa & \kappa & \kappa & \kappa
 \end{pmatrix}$ & [$7,2,4$] \\

 $\begin{pmatrix}
     & \kappa & \kappa & 0 & \kappa \\
     \kappa I_3 & \kappa & \kappa & 0 & 0 \\
     & \kappa & 0 & \kappa & \kappa
 \end{pmatrix}$ & [$7,3,3$] \\

 $\begin{pmatrix}
     & 0 & \kappa & 0 \\
     \kappa I_4 & \kappa & 0 & \kappa \\
     & 0 & \kappa & 0 \\
     & 0 & \kappa & 0
 \end{pmatrix}$ & [$7,4,2$] \\

 $\begin{pmatrix}
     & 0 & \kappa & 0 \\
     \kappa I_4 & \kappa & 0 & \kappa \\
     & 0 & \kappa & 0 \\
     & 0 & \kappa & \kappa
 \end{pmatrix}$ & [$7,4,2$] \\

 $\begin{pmatrix}
     & \kappa & \kappa & 0 \\
     \kappa I_4 & \kappa & 0 & \kappa \\
     & 0 & \kappa & 0 \\
     & 0 & \kappa & 0
 \end{pmatrix}$ & [$7,4,2$] \\

 $\begin{pmatrix}
     & \kappa & \kappa & 0 \\
     \kappa I_4 & \kappa & 0 & \kappa \\
     & \kappa & 0 & 0 \\
     & \kappa & 0 & 0
 \end{pmatrix}$ & [$7,4,2$] \\

\hline
\end{tabular}
\end{minipage}
\end{table}

\begin{table}[ht]
\centering
\begin{minipage}{0.4\textwidth}
\centering
\caption{\label{Tab4} $l_1$-optimal free codes (Part III).}

\begin{tabular}{|c|c|}
\hline
Generator Matrix & $[n,k,d]$ \\ \hline
  
$\begin{pmatrix}
     & 0 & \kappa & \kappa \\
     \kappa I_4 & \kappa & \kappa & \kappa \\
     & 0 & \kappa & 0 \\
     & \kappa & \kappa & 0
 \end{pmatrix}$ & [$7,4,2$] \\

  $\begin{pmatrix}
     & 0 & \kappa & 0 \\
     \kappa I_4 & \kappa & 0 & \kappa \\
     & 0 & 0 &\kappa \\
     & 0 & 0 &\kappa
 \end{pmatrix}$ & [$7,4,2$] \\ 

$\begin{pmatrix}
      & 0 & \kappa \\
      & 0 & \kappa \\ 
      \kappa I_5 & 0 & \kappa \\
      & 0 & \kappa \\ 
      & 0 & \kappa \\
  \end{pmatrix}$ &[$7,5,2$] \\
  $\begin{pmatrix}
      & \kappa & 0 \\
      & 0 & \kappa \\ 
      \kappa I_5 & \kappa & 0 \\
      & \kappa & 0 \\ 
      & 0 & \kappa \\
  \end{pmatrix}$ &[$7,5,2$] \\

   $\begin{pmatrix}
      & \kappa & 0 \\
      & 0 & \kappa \\ 
      \kappa I_5 & 0 & \kappa \\
      & 0 & \kappa \\ 
      & 0 & \kappa \\
  \end{pmatrix}$ &[$7,5,2$] \\

  $\begin{pmatrix}
      & \kappa & \kappa \\
      & 0 & \kappa \\ 
      \kappa I_5 & \kappa & \kappa \\
      & 0 & \kappa \\ 
      & 0 & \kappa \\
  \end{pmatrix}$ &[$7,5,2$] \\

   $\begin{pmatrix}
      & \kappa & \kappa \\
      & \kappa & 0 \\ 
      \kappa I_5 & \kappa & \kappa \\
      & 0 & \kappa \\ 
      & 0 & \kappa \\
  \end{pmatrix}$ &[$7,5,2$] \\

  $\begin{pmatrix}
      & 0 \\
      & 0 \\ 
     \kappa I_6 & 0 \\
      & 0 \\
      & 0 \\
      & \kappa
  \end{pmatrix}$ & [$7,6,1$] \\

  $\begin{pmatrix}
      & \kappa \\
      & \kappa \\ 
     \kappa I_6 &\kappa \\
      & \kappa \\
      & \kappa \\
      & 0
  \end{pmatrix}$ & [$7,6,1$] \\

\hline
\end{tabular}

\end{minipage}
\hspace{0.05\textwidth}
\begin{minipage}{0.45\textwidth}
\centering
\caption{\label{Tab5}  $l_1$-optimal free codes (Part IV).}

\begin{tabular}{|c|c|}
\hline

Generator Matrix & [$n,k,d$] \\
\hline

$\begin{pmatrix}
      & \kappa \\
      & 0 \\ 
     \kappa I_6 &\kappa \\
      & 0 \\
      & \kappa \\
      & 0
  \end{pmatrix}$ & [$7,6,1$] \\

$\begin{pmatrix}
    \kappa & \kappa & \kappa &  \kappa & \kappa & \kappa & \kappa & \kappa &
\end{pmatrix}$ & [$8,1,8$] \\
$\begin{pmatrix}
    \kappa & 0 & \kappa & 0 & \kappa & \kappa & 0 & 0 \\
    0 & \kappa & \kappa & \kappa & 0 & \kappa & 0 & \kappa
\end{pmatrix}$ & [$8,2,4$] \\

$\begin{pmatrix}
    & \kappa & \kappa & \kappa & 0 & 0 \\
    \kappa I_3 & \kappa & \kappa & 0 & \kappa & 0 \\
    & 0 & \kappa & \kappa & 0 & \kappa 
\end{pmatrix}$ & [$8,3,4$] \\

$\begin{pmatrix}
    & \kappa & 0 & 0 & \kappa \\
    \kappa I_4 & \kappa & \kappa & 0 & 0 \\
    & \kappa & 0 & 0 & \kappa \\
    & \kappa & 0 & \kappa & 0 
\end{pmatrix}$ & [$8,4,3$] \\ 
$\begin{pmatrix}
    & \kappa & 0 & \kappa \\
    & 0 & \kappa & \kappa \\
    \kappa I_5 & 0 & 0 & \kappa \\
    & 0 & 0 & \kappa \\ 
    & 0 & 0 & \kappa 
\end{pmatrix}$ & [$8,5,2$] \\

$\begin{pmatrix}
    & \kappa & 0 & \kappa \\
    & 0 & \kappa & 0\\
    \kappa I_5 & 0 & \kappa & 0 \\
    & 0 & \kappa & \kappa \\ 
    & 0 & 0 & \kappa 
\end{pmatrix}$ & [$8,5,2$] \\

$\begin{pmatrix}
    & 0 &\kappa & \kappa \\
    & \kappa & 0 & \kappa \\
    \kappa I_5 & 0 & \kappa & \kappa \\
    & 0 & \kappa & 0 \\ 
    & 0 & 0 & \kappa 
\end{pmatrix}$ & [$8,5,2$] \\

$\begin{pmatrix}
    & \kappa & 0 & 0 \\
    & 0 & 0 & \kappa \\
    \kappa I_5 & 0 & \kappa& 0 \\
    & 0 & \kappa& 0 \\ 
    & 0 & 0 & \kappa 
\end{pmatrix}$ & [$8,5,2$] \\
$\begin{pmatrix}
    & \kappa & 0 & 0 \\
    & \kappa & \kappa & \kappa \\
    \kappa I_5 & 0 & \kappa& 0 \\
    & 0 & \kappa& 0 \\ 
    & 0 & 0 & \kappa 
\end{pmatrix}$ & [$8,5,2$] \\

$\begin{pmatrix}
    & \kappa & 0 & \kappa \\
    & 0 &\kappa & 0 \\
    \kappa I_5 & 0 & 0& \kappa \\
    & \kappa & \kappa& \kappa \\ 
    & \kappa & \kappa & 0 
\end{pmatrix}$ & [$8,5,2$] \\

\hline
\end{tabular}
\end{minipage}
\end{table}

\begin{table}[ht]
\centering
\begin{minipage}{0.4\textwidth}
\centering
\caption{\label{Tab6} $l_1$-optimal free codes (Part V).}
\begin{tabular}{|c|c|}
\hline
Generator Matrix & $[n,k,d]$ \\ \hline

$\begin{pmatrix}
    & 0 & \kappa & \kappa \\
    & 0 & \kappa & 0 \\
    \kappa I_5 & 0 & \kappa& \kappa \\
    & 0 & 0& \kappa \\ 
    & 0 & 0 & \kappa 
\end{pmatrix}$ & [$8,5,2$] \\

$\begin{pmatrix}
    & \kappa & 0 & 0 \\
    & 0 & \kappa & \kappa \\
    \kappa I_5 & 0 & \kappa& 0\\
    & 0 & \kappa& 0 \\ 
    & 0 & 0 & \kappa 
\end{pmatrix}$ & [$8,5,2$] \\
$\begin{pmatrix}
    & 0 & \kappa & \kappa \\
    & 0 & 0 & \kappa \\
    \kappa I_5 & 0 & \kappa& \kappa \\
    & 0 & 0& \kappa \\ 
    & 0 & 0 & \kappa 
\end{pmatrix}$ & [$8,5,2$] \\

$\begin{pmatrix}
    & \kappa & \kappa & 0 \\
    &  \kappa & 0 & \kappa \\
    \kappa I_5  & \kappa& 0 & 0 \\
    &  \kappa & \kappa & 0\\ 
    & 0 & 0 & \kappa 
\end{pmatrix}$ & [$8,5,2$] \\

$\begin{pmatrix}
    & \kappa & \kappa & 0 \\
     & 0 & \kappa & \kappa \\
    \kappa I_5  & \kappa& 0 & 0 \\
    &  \kappa & \kappa & 0\\ 
    & 0 & 0 & \kappa 
\end{pmatrix}$ & [$8,5,2$] \\

$\begin{pmatrix}
    & 0 & \kappa & \kappa  \\
     & 0 & 0 & \kappa \\
    \kappa I_5  & \kappa& 0 & 0 \\
    &  \kappa & 0 & 0\\ 
    & 0 & \kappa & \kappa 
\end{pmatrix}$ & [$8,5,2$] \\ 

$\begin{pmatrix}
    & \kappa & \kappa & \kappa  \\
     & 0  & \kappa & 0 \\
    \kappa I_5  & \kappa& 0 & 0 \\
    &   0 & \kappa & 0\\ 
    & 0 & 0 & \kappa 
\end{pmatrix}$ & [$8,5,2$] \\

$\begin{pmatrix}
    & 0 & \kappa & 0  \\
     & \kappa & 0 & \kappa \\
    \kappa I_5  & 0& \kappa & 0 \\
    &  0& \kappa & 0\\ 
    & 0 & 0 & \kappa 
\end{pmatrix}$ & [$8,5,2$] \\

$\begin{pmatrix}
    & 0 & \kappa & 0  \\
     & 0 & 0 & \kappa \\
    \kappa I_5  & 0& \kappa & 0 \\
    &  0& \kappa & 0\\ 
    & 0 & 0 & \kappa 
\end{pmatrix}$ & [$8,5,2$] \\

\hline
\end{tabular}

\end{minipage}
\hspace{0.05\textwidth}
\begin{minipage}{0.45\textwidth}
\centering
\caption{\label{Tab7} $l_1$-optimal free codes (Part VI).}
\begin{tabular}{|c|c|}
\hline
Generator Matrix & $[n,k,d]$ \\ \hline
  
  $\begin{pmatrix}
    & 0 & 0 & \kappa  \\
     & 0 & 0 & \kappa \\
    \kappa I_5  & 0& 0 & \kappa \\
    &  0& 0 & \kappa \\ 
    & 0 & 0 & \kappa 
\end{pmatrix}$ & [$8,5,2$] \\ 

$\begin{pmatrix}
    & 0 & 0 & \kappa  \\
     & \kappa & \kappa & \kappa \\
    \kappa I_5  & 0& 0 & \kappa \\
    &  0& 0 & \kappa \\ 
    & 0 & 0 & \kappa 
\end{pmatrix}$ & [$8,5,2$] \\

  $\begin{pmatrix}
    & 0 & \kappa & 0  \\
     & \kappa & 0 & \kappa \\
    \kappa I_5  & 0& 0 & \kappa \\
    &  0& 0 & \kappa \\ 
    & 0 & 0 & \kappa 
\end{pmatrix}$ & [$8,5,2$] \\
  
  $\begin{pmatrix}
    & 0 & \kappa & 0  \\
     & 0 & 0 & \kappa \\
    \kappa I_5  & 0& 0 & \kappa \\
    &  0& 0 & \kappa \\ 
    & 0 & 0 & \kappa 
\end{pmatrix}$ & [$8,5,2$] \\
$\begin{pmatrix}
    & 0 & \kappa & 0  \\
     & \kappa & \kappa & \kappa \\
    \kappa I_5  & 0& 0 & \kappa \\
    &  0& 0 & \kappa \\ 
    & 0 & 0 & \kappa 
\end{pmatrix}$ & [$8,5,2$] \\  

$\begin{pmatrix}
    & 0 & \kappa \\
    & \kappa & \kappa \\
    \kappa I_6 & 0 & 0 \\
    & 0 & 0 \\ 
    & 0 & 0 \\ 
    & \kappa & \kappa 
\end{pmatrix}$ & [$8,6,1$] \\
$\begin{pmatrix}
    & 0 & \kappa \\
    & 0 & 0 \\
    \kappa I_6 & 0 & 0 \\
    & 0 & 0 \\ 
    & 0 & 0 \\ 
    & 0 & 0 
\end{pmatrix}$ & [$8,6,1$] \\
$\begin{pmatrix}
    & 0 & \kappa \\
    & 0 & \kappa \\
    \kappa I_6 & 0 & \kappa \\
    & 0 & \kappa \\ 
    & 0 &\kappa \\ 
    & 0 & 0 
\end{pmatrix}$ & [$8,6,1$] \\

\hline
\end{tabular}
\end{minipage}
\end{table}

 \begin{table}[ht]
\centering
\begin{minipage}{0.4\textwidth}
\centering
\caption{\label{Tab8a} $l_1$-optimal free codes (Part VII).}
\begin{tabular}{|c|c|}
\hline
Generator Matrix & $[n,k,d]$ \\ \hline
$\begin{pmatrix}
    & \kappa & 0 \\
    & 0 & \kappa \\
    \kappa I_6 & 0 & \kappa \\
    & 0 & \kappa \\ 
    & 0 &0 \\ 
    & 0 & \kappa 
\end{pmatrix}$ & [$8,6,1$] \\
$\begin{pmatrix}
    & 0 & \kappa \\
    & \kappa & \kappa \\
    \kappa I_6 & 0 & 0 \\
    & 0 & \kappa \\ 
    & 0 & \kappa \\ 
    & \kappa & 0 
\end{pmatrix}$ & [$8,6,1$] \\

$\begin{pmatrix}
    & 0 & \kappa \\
    & 0 & \kappa \\
    \kappa I_6 & 0 & 0 \\
    & 0 & \kappa \\ 
    & 0 &0 \\ 
    & 0 & 0 
\end{pmatrix}$ & [$8,6,1$] \\

$\begin{pmatrix}
    & \kappa &\kappa \\
    & 0 & \kappa \\
    \kappa I_6 & \kappa & 0 \\
    & 0 & \kappa \\ 
    & \kappa &\kappa \\ 
    & 0 & 0 
\end{pmatrix}$ & [$8,6,1$] \\

$\begin{pmatrix}
    & 0 & \kappa \\
    & \kappa & 0 \\
    \kappa I_6 & 0 & \kappa \\
    & 0 & \kappa \\ 
    & 0 &0 \\ 
    & \kappa & 0 
\end{pmatrix}$ & [$8,6,1$] \\

$\begin{pmatrix}
    & \kappa & 0 \\
    & 0 & \kappa \\
    \kappa I_6 & 0 & \kappa \\
    & 0 & 0 \\ 
    & 0 &\kappa \\ 
    & 0 & 0 
\end{pmatrix}$ & [$8,6,1$] \\

$\begin{pmatrix}
    & \kappa \\
    & \kappa \\ 
    & \kappa \\
    \kappa I_6& \kappa \\
    & \kappa \\
    & \kappa \\ 
    & \kappa \\
\end{pmatrix}$ & [$8,7,2$] \\
\hline

\end{tabular}

\end{minipage}
\hspace{0.05\textwidth}
\begin{minipage}{0.45\textwidth}
\centering
\caption{\label{Tab8} $l_2$-optimal free codes (Part I).}
\begin{tabular}{|c|c|}
\hline
Generator Matrix & $[n,k,d]$ \\ \hline

$\begin{pmatrix}
    \kappa & 0 & \kappa & 0 \\
    0 & \kappa & 0 & \kappa 
\end{pmatrix}$ & [$4,2,2$]  \\ 
$\begin{pmatrix}
    \kappa & 0 & 0 & \kappa & 0 \\
    0 & \kappa & \kappa & 0 & 0 
\end{pmatrix}$ & [$5,2,2$] \\ 

$\begin{pmatrix}
      & 0 & \kappa \\
     \kappa I_3 & \kappa & 0 \\
     &  0 & 0
\end{pmatrix}$ & [$5,3,1$] \\ 

$\begin{pmatrix}
    \kappa & 0 & \kappa & 0 & \kappa & \kappa \\
    0 & \kappa & \kappa & \kappa & 0 & \kappa
\end{pmatrix}$ & [$6,2,4$] \\ 

$\begin{pmatrix}
    & \kappa & \kappa & \kappa \\
    \kappa I_3 & \kappa & 0 & \kappa \\
    & 0 & \kappa & \kappa 
\end{pmatrix}$ & [$6,3,3$] \\

$\begin{pmatrix}
    & 0 & \kappa \\
    \kappa I_4 & \kappa & 0 \\
    & \kappa & 0 \\
    & \kappa & 0
\end{pmatrix}$ & [$6,4,2$]\\ 

$\begin{pmatrix}
    & \kappa & 0 \\
    \kappa I_4 & \kappa & \kappa \\
    & 0 & \kappa \\
    & \kappa & \kappa
\end{pmatrix}$ & [$6,4,2$] \\ 

$\begin{pmatrix}
    \kappa & 0 & 0 & \kappa & \kappa & 0 & \kappa \\
    0 & \kappa & \kappa & \kappa & \kappa & 0 & 0 
\end{pmatrix}$ & [$7,2,4$] \\

$\begin{pmatrix}
    & 0 & \kappa & 0 & \kappa \\
    \kappa I_3 & 0 & \kappa & \kappa & \kappa \\
    & 0 & \kappa & \kappa & 0
\end{pmatrix}$ & [$7,3,3$] \\

$\begin{pmatrix}
    & 0 & \kappa & 0 & \kappa \\
    \kappa I_3 & \kappa & \kappa & 0 & 0 \\
    & 0 & \kappa & \kappa & 0
\end{pmatrix}$ & [$7,3,3$] \\

$\begin{pmatrix}
    & 0 & 0 & \kappa \\
    \kappa I_4 & 0 & \kappa & \kappa \\
    & 0 & \kappa & \kappa \\
    & 0 & \kappa & 0 
\end{pmatrix}$ & [$7,4,2$] \\

$\begin{pmatrix}
    & 0 & 0 & \kappa \\
    \kappa I_4 & 0 & \kappa & 0 \\
    & \kappa & 0 & 0 \\
    & \kappa & 0 & 0 
\end{pmatrix}$ & [$7,4,2$] \\ 

$\begin{pmatrix}
    & 0 & 0 & \kappa \\
    \kappa I_4 & \kappa & \kappa & \kappa \\
    & 0 & \kappa & 0 \\
    & \kappa & \kappa &\kappa 
\end{pmatrix}$ & [$7,4,2$] \\ 

$\begin{pmatrix}
    & 0 & 0 & \kappa \\
    \kappa I_4 & 0 & 0 & 0 \\
    & \kappa & \kappa & 0 \\
    & \kappa & \kappa & 0 
\end{pmatrix}$ & [$7,4,2$] \\

$\begin{pmatrix}
    & 0 & 0 & \kappa \\
    \kappa I_4 & 0 & \kappa & 0 \\
    & 0 & 0 & \kappa \\
     & 0 & 0 & \kappa 
\end{pmatrix}$ & [$7,4,2$] \\

\hline
\end{tabular}

\end{minipage}
\end{table}

 \begin{table}[ht]
\centering
\begin{minipage}{0.4\textwidth}
\centering
\caption{\label{Tab9}  $l_2$-optimal free codes (Part II).}
\begin{tabular}{|c|c|}
\hline
Generator Matrix & [$n,k,d$] \\
\hline

$\begin{pmatrix}
    & 0 & \kappa \\
    & \kappa & 0 \\
    \kappa I_5 & \kappa & 0 \\
    & \kappa & 0 \\
    & 0 & 0
\end{pmatrix}$ & [$7,5,1$] \\
$\begin{pmatrix}
    & 0 & \kappa \\
    & \kappa & 0 \\
    \kappa I_5 & 0 & 0 \\
    & 0 & 0 \\
    & 0 & 0
\end{pmatrix}$ & [$7,5,1$] \\

$\begin{pmatrix}
    & 0 & \kappa \\
    & \kappa & \kappa \\
    \kappa I_5 & \kappa & \kappa \\
    & \kappa & 0 \\
    & 0 & 0
\end{pmatrix}$ & [$7,5,1$] \\
$\begin{pmatrix}
    \kappa & 0 & 0 & \kappa & \kappa & 0 & \kappa & 0 \\
    0 & \kappa & \kappa & \kappa & \kappa & 0 & 0 & 0
\end{pmatrix}$ & [$8,2,4$] \\

$\begin{pmatrix}
    \kappa & 0 & 0 & \kappa & \kappa & \kappa & \kappa & \kappa \\
    0 & \kappa & \kappa & \kappa & \kappa & 0 & 0 & 0
\end{pmatrix}$ & [$8,2,4$] \\

$\begin{pmatrix}
    \kappa & 0 & 0 & 0 & 0 & \kappa & \kappa & \kappa \\
    0 & \kappa & \kappa & \kappa & \kappa & 0 & 0 & 0
\end{pmatrix}$ & [$8,2,4$] \\

$\begin{pmatrix}
    & \kappa & \kappa & 0 & \kappa & 0\\
    \kappa I_3 & \kappa & 0 & \kappa & \kappa & \kappa \\
    & 0 & \kappa & \kappa & \kappa & \kappa
\end{pmatrix}$ & [$8,3,4$] \\

$\begin{pmatrix}
    & \kappa & \kappa & \kappa& 0 \\
    \kappa I_4 & \kappa & \kappa & 0 & \kappa \\
    & 0 & \kappa & \kappa & \kappa \\
   & \kappa & \kappa & 0 & 0 \\
\end{pmatrix}$ & [$8,4,3$] \\

$\begin{pmatrix}
    & 0 & \kappa & \kappa& \kappa \\
    \kappa I_4 & 0 & \kappa & 0 & \kappa \\
    & \kappa & \kappa & \kappa & 0 \\
     & 0 & 0 & \kappa & \kappa\\
\end{pmatrix}$ & [$8,4,3$] \\ 

$\begin{pmatrix}
    & 0 & \kappa & \kappa \\
    & \kappa & 0 & \kappa \\
    \kappa I_5 & 0 & 0 & \kappa \\
    & \kappa & \kappa & 0 \\
    & 0 & 0 & \kappa 
\end{pmatrix}$ & [$8,5,2$] \\

$\begin{pmatrix}
    & 0 & \kappa & \kappa \\
    & \kappa  & \kappa & 0 \\
    \kappa I_5 & \kappa & \kappa & \kappa \\
    & \kappa  & 0 & \kappa\\
    & 0 & 0 & \kappa 
\end{pmatrix}$ & [$8,5,2$] \\

$\begin{pmatrix}
    & 0 & \kappa\\
    & \kappa & \kappa \\
   \kappa I_6 & \kappa & \kappa \\
    &\kappa & 0 \\
    & 0 & \kappa \\
    & 0 & \kappa
\end{pmatrix}$ & [$8,6,2$] \\

\hline
\end{tabular}

\end{minipage}
\hspace{0.12\textwidth}
\begin{minipage}{0.45\textwidth}

\caption{\label{Tab9a}  $l_2$-optimal free codes (Part III).}
\begin{tabular}{|c|c|}
\hline
Generator Matrix & [$n,k,d$] \\
\hline

\small$\begin{pmatrix}
    & 0 & \kappa\\
    & \kappa & 0 \\
   \kappa I_6 & \kappa & 0 \\
    &\kappa & 0 \\
    & 0 & \kappa \\
    & 0 & \kappa
\end{pmatrix}$ & [$8,6,2$] \\

\small$\begin{pmatrix}
    & 0 & \kappa\\
    & \kappa & \kappa \\
   \kappa I_6 & 0 & \kappa \\
    &0 & \kappa \\
    & 0 & \kappa \\
    & 0 & \kappa
\end{pmatrix}$ & [$8,6,2$] \\

\hline
\end{tabular}
\vspace{1.5cm}
\centering
\caption{\label{Tab10} $l_3$-optimal free codes (Part I).}
\begin{tabular}{|c|c|}
\hline
Generator Matrix & $[n,k,d]$ \\ \hline
$\begin{pmatrix}
      & 0 & 0 & \kappa \\
      \kappa I_3 & \kappa & 0 & 0\\
      & 0 & \kappa & 0
  \end{pmatrix}$ & [$6,3,2$] \\
  
$\begin{pmatrix}
    & 0 & \kappa & \kappa & \kappa \\
    \kappa I_3 & \kappa & \kappa & 0 & \kappa \\
    & \kappa & 0 & \kappa & \kappa
\end{pmatrix}$ &  [$7,3,4$]  \\ 

$\begin{pmatrix}
    & \kappa & \kappa & \kappa \\
    \kappa I_4 & 0 & \kappa & \kappa \\
    & \kappa & 0 & \kappa \\
    & \kappa & \kappa & 0
\end{pmatrix}$ & [$7,4,3$] \\

$\begin{pmatrix}
    & \kappa & 0 & \kappa & \kappa & 0\\
   \kappa I_3 & \kappa & \kappa & \kappa & 0 & 0 \\
    & \kappa & 0 & \kappa & 0 & \kappa 
\end{pmatrix}$ & [$8,3,4$] \\

$\begin{pmatrix}
    & \kappa & \kappa & 0 & 0 & \kappa \\
   \kappa I_3 & \kappa & 0 & \kappa & 0 & \kappa \\
    & \kappa & \kappa & \kappa & 0 & 0 
\end{pmatrix}$ & [$8,3,4$] \\

$\begin{pmatrix}
  & \kappa & 0 & 0 & \kappa \\
  \kappa I_4 & 0 & \kappa & 0 & \kappa \\
  & \kappa & \kappa & 0 & 0 \\
  & \kappa & \kappa & 0 & \kappa
\end{pmatrix}$ & [$8,4,3$] \\

\hline
\end{tabular}

\end{minipage}
\end{table}

\begin{table}[ht]
\centering
\begin{minipage}{0.4\textwidth}
\centering
\caption{\label{Tab10a} $l_3$-optimal free codes (Part II).}
\begin{tabular}{|c|c|}
\hline
Generator Matrix & [$n,k,d$] \\
\hline
$\begin{pmatrix}
    & 0 & \kappa & 0 \\
    & 0 & 0 & \kappa \\
    \kappa I_5 & \kappa & 0 & 0 \\
    & 0 & \kappa & 0 \\
    & 0 & \kappa & 0
\end{pmatrix}$ & [$8,5,2$] \\

$\begin{pmatrix}
    & \kappa & \kappa & 0 \\
    & 0 & 0 & \kappa \\
    \kappa I_5 & \kappa & 0 & 0 \\
    & \kappa & \kappa & 0 \\
    & 0 & \kappa & 0
\end{pmatrix}$ & [$8,5,2$] \\

$\begin{pmatrix}
    & 0 & \kappa & 0 \\
    & \kappa &\kappa & \kappa \\
    \kappa I_5 & 0 & 0 & \kappa \\
     & \kappa & 0 & 0 \\
    & \kappa & \kappa & \kappa
\end{pmatrix}$ & [$8,5,2$] \\

\hline
\end{tabular}

\end{minipage}
\hspace{0.12\textwidth}
\begin{minipage}{0.45\textwidth}
\centering
\caption{\label{Tab11}  $l_4$-optimal free codes.}
\begin{tabular}{|c|c|}
\hline
Generator Matrix & [$n,k,d$] \\
\hline

$\begin{pmatrix}
    & \kappa & \kappa & 0 & \kappa \\
    \kappa I_4 & 0 & \kappa & \kappa & \kappa \\
    & \kappa & \kappa & \kappa & 0 \\
    & \kappa & 0 & \kappa & \kappa 
\end{pmatrix}$ & [$8,4,4$] \\
\hline
\end{tabular}
\end{minipage}
\end{table}

\section{Conclusion}
In this paper, we have studied hulls of free linear codes over a non-unital ring $E$. Initially, we have focused on residue and torsion codes of various hulls and found the generator matrix of the hull of a free $E$-linear code. Then, four build-up construction methods have been given for constructing free $E$-linear codes with a larger length and hull-rank from free $E$-linear codes with a smaller length and hull-rank. Some examples of codes constructed by these build-up construction methods are also given. Later, we studied the permutation equivalence of two free $E$-linear codes and then discussed the hull-variation problem. Finally, we have classified optimal free $E$-linear codes under permutation equivalence for lengths up to $8$. We would also like to point out that extending this work to other non-unital rings, as appeared in the classification of Fine \cite{Fine93}, can be an important direction for future research.

\section*{Acknowledgement}
The first author gratefully acknowledges financial support from the Council of Scientific \& Industrial Research, Govt. of India (under grant No. 09/1023(16098)/2022-EMR-I).
%Also, the authors would like to thank the anonymous referee(s) and the Editor of this journal for their valuable comments to improve the presentation of the paper.
\section*{Declarations}
\textbf{Competing interests}: All the authors confirm that there is no competing interests related to this manuscript.\\
\textbf{Statement on Data Availability}: The authors confirm that this manuscript encompasses all the data used to support the findings of this study. For any essential clarifications, requests can be directed to the corresponding author. \\
\textbf{Use of AI tools}: No Artificial Intelligence (AI) tools were utilized in the writing or preparation of this manuscript.


\begin{thebibliography}{}

\bibitem{Alah21}
 A. Alahmadi,  A. Altassan, W. Basaffar, A. Bonnecaze, H. Shoaib and  P. Sol\'e, Quasi type IV codes over a non-unital ring, Appl. Algebra Engrg. Comm. Comput., \textbf{32}(3) (2021), 217-228.

\bibitem{Alah22}
A. Alahmadi,  A. Alkathiry, A. Altassan, A. Bonnecaze, H. Shoaib and  P. Sol\'e, The build-up construction over a commutative non-unital ring, Des. Codes Cryptogr., \textbf{90}(12) (2022), 3003-3010.

\bibitem{Alah23}
A. Alahmadi, A. Melaibari and  P. Sol\'e, Duality of codes over non-unital rings of order four, IEEE Access, \textbf{11} (2023), 53120-53133.

\bibitem{Assmus}
E.F. Assmus Jr. and J.D. Key, Affine and projective planes, Discret. Math., \textbf{83}(2-3) (1990), 161–187.

\bibitem{Magma}
W. Bosma and J. Cannon, Handbook of Magma Functions, Univ. of Sydney, (1995)

\bibitem{Bachoc}
C. Bachoc and P. Gaborit, Designs and self-dual codes with long shadows, J. Combinat. Theory A, \textbf{105}(1) (2004),  15-34.

\bibitem{Bannai}
E. Bannai, S.T. Dougherty, M. Harada and M. Oura, Type II codes, even unimodular lattices, and invariant rings, IEEE Trans. Inf. Theory,
\textbf{45}(4) (1999), 1194-1205.

\bibitem{Carl16}
C. Carlet and S. Guilley, Complementary dual codes for counter-measures to side-channel attacks, Adv. Math. Commun., \textbf{10}(1) (2016), 131-150.

\bibitem{Conway}
J.H. Conway and N.J.A. Sloane, Sphere Packings, Lattices and Groups, 3rd ed. New York, NY, USA: Springer, 1998.

\bibitem{Calderbank}
A.R. Calderbank and P.W. Shor, Good quantum error-correcting codes exist, Phys. Rev. A, Gen. Phys.,  \textbf{54}(2) (1996), 1098-1105.

 \bibitem{Deb}
 S. Deb, I. Kikani and M.K. Gupta, On the classification of codes over non-unital ring of order $4$, Discrete Math. Algorithms Appl., \textbf{16}(6) (2024), 2350076.

\bibitem{Dougherty}
S.T. Dougherty and S. Sahinkaya, On cyclic and negacyclic codes with one-dimensional hulls and their applications, Adv. Math. Commun., \textbf{18} (2024), 1364-1378.




\bibitem{Fine93}
 B. Fine, Classification of finite rings of order $p^2$, Math. Mag., \textbf{66}(4) (1993), 248-252.




\bibitem{Harada}
M. Harada, On the existence of frames of the Niemeier lattices and self-dual codes over $\mathbb{F}_p$, J. Algebra, \textbf{321}(8) (2009), 2345–2352.

\bibitem{Islam2022}
H. Islam and O. Prakash, Construction of LCD and new quantum codes from cyclic codes over a finite non-chain ring, Cryptogr. Commun., \textbf{14}(1) (2022), 59-73.

\bibitem{Kim2022a}
J.L. Kim and Y.G. Roe, Construction of quasi-self-dual codes over a commutative non-unital ring of order $4$, Appl. Algebra Engrg. Comm. Comput., \textbf{35}(3) (2022), 393-406.

\bibitem{Kushwaha1}
A. Kushwaha, S. Yadav and O. Prakash, Quasi-self-dual, Right LCD and ACD codes over a noncommutative non-unital ring, J. Algebra Comb. Discrete Struct. Appl., \textbf{11}(3) (2024), 207-225.

\bibitem{Kushwaha2} 
 A. Kushwaha, I. Debnath and O. Prakash, Classification of LCD and self-dual codes over a finite non-unital local ring, 2025. https://arxiv.org/pdf/2501.03016v1


\bibitem{Li18}
C. Li and P. Zeng, Constructions of linear codes with one-dimensional hull, IEEE Trans. Inform. Theory, \textbf{65}(3) (2019), 1668-1676.


\bibitem{Leon}
J.S. Leon, Permutation group algorithms based on partitions, I: theory and algorithms, J. Symbolic Comput., \textbf{12}(4-5) (1991), 533-583.


\bibitem{Liu}
H. Liu and X. Pan, Galois hulls of linear codes over finite fields, Des. Codes Cryptogr., \textbf{88}(2) (2020), 241-255.

\bibitem{Liu2}
  R. Liu, S. Li and M. Shi, Construction of linear codes with various Hermitian hull dimensions and related EAQECCs, Adv. Math. Commun., \textbf{19}(2)  (2025), 588-603.

\bibitem{Li}
 S. Li and M. Shi, Characterization and classification of binary linear codes with various hull dimensions from an improved mass formula, IEEE Trans. Inf. Theory, \textbf{70}(5) (2023), 3357-3372.

\bibitem{Massey92}
J.L. Massey, Linear codes with complementary duals, Discrete Math., \textbf{106} (1992), 337-342.



\bibitem{Prakash} 
O. Prakash, S. Yadav, H. Islam and P. Sol\'e, Self-dual and LCD double circulant codes over a class of non-local rings, Comput. Appl. Math., \textbf{41}(6)  (2022), 1-16.


\bibitem{Shi21}
M. Shi, S. Li, J.L. Kim and P. Sol\'e, LCD and ACD codes over a noncommutative non-unital ring with four elements, Cryptogr. Commun., \textbf{14}(3) (2022), 627-640.

  
\bibitem{Send1997}
N. Sendrier, Finding the permutation between equivalent binary codes, In: Proceedings of IEEE ISIT, Ulm, Germany, pp. 367, (1997).
%https://doi. org/10. 1109/ISIT.1997.

\bibitem{Send2000}
N. Sendrier, Finding the permutation between equivalent codes: the support splitting algorithm, IEEE Trans. Inform. Theory, \textbf{46}(4) (2000), 1193-1203.


\bibitem{Send04}
 N. Sendrier, Linear codes with complementary duals meet the Gilbert–Varshamov bound, Discrete Math., \textbf{285}(1) (2004), 345–347.

\bibitem{Minjia}
M. Shi, Y.J. Choie, A. Sharma and P. Sol\'e, Codes and Modular Forms: A Dictionary. Singapore: World Scientific, 2019.

\bibitem{Steane}
A.M. Steane, Error-correcting codes in quantum theory, Phys. Rev. Lett., \textbf{77}(5) (1996), 793-797.

 \bibitem{Wang}
Y. Wang and R. Tao, Constructions of linear codes with small hulls from association schemes, Adv. Math. Commun., \textbf{16}(2) (2022), 349-364.

\bibitem{Yadav2025}
 S. Yadav, I. Debnath and O. Prakash, Some constructions of $l$-Galois LCD codes, Adv. Math. Commun., \textbf{19}(1)  (2025), 227–244.




\end{thebibliography}
\end{document}